\documentclass[10pt, conference]{IEEEtran}
\usepackage{cite}
\usepackage{amsthm}
\usepackage{amsmath,amssymb,amsfonts}
\usepackage{algorithmic}
\usepackage{graphicx}
\usepackage{textcomp}
\usepackage{xcolor}
\usepackage{array}
\usepackage{url}
\usepackage[colorlinks,breaklinks=true,bookmarks=true]{hyperref}
\usepackage{mymacros}
\usepackage{enumitem, lipsum}
\usepackage{xspace}

\def\BibTeX{{\rm B\kern-.05em{\sc i\kern-.025em b}\kern-.08em
    T\kern-.1667em\lower.7ex\hbox{E}\kern-.125emX}}

\newtheoremstyle{lemmastyle}
{\topsep} 
{\topsep} 
{} 
{} 
{\bfseries} 
{.} 
{0.5em} 
{} 

\theoremstyle{plain}
\newtheorem{definitions}{Definition}
\newtheorem{lemma}{Lemma}
\newtheorem{assumption}{Assumption}
\newtheorem{theorem}{Theorem}

\newif{\ifanonymous}
\anonymousfalse
 
\title{Relational Analysis of Sensor Attacks on Cyber-Physical Systems}
\ifanonymous
\else
\author{
  \IEEEauthorblockN{Jian Xiang\IEEEauthorrefmark{1}, Nathan Fulton\IEEEauthorrefmark{2}, Stephen Chong\IEEEauthorrefmark{1}}
  \IEEEauthorblockA{\IEEEauthorrefmark{1}SEAS, Harvard University \emph{\{jxiang, chong\}@seas.harvard.edu}}
\IEEEauthorblockA{\IEEEauthorrefmark{2}MIT-IBM Watson AI Lab. \emph{nathan@ibm.com}}
}
\fi

\begin{document}
\bstctlcite{IEEEexample:BSTcontrol}

\maketitle

\begin{abstract}

Cyber-physical systems, such as self-driving cars or autonomous aircraft, must defend against attacks that target sensor hardware.
Analyzing system design can help engineers understand how a compromised sensor could impact the system's behavior;
however, designing security analyses for cyber-physical systems is difficult due to their combination of discrete dynamics, continuous dynamics, and nondeterminism.

This paper contributes a framework for modeling and analyzing sensor attacks on cyber-physical systems, using the formalism of hybrid programs.
We formalize and analyze two relational properties of a system's robustness. These relational properties respectively express (1) whether a system's safety property can be influenced by sensor attacks, and (2) whether a system's high-integrity state can be affected by sensor attacks. We characterize these relational properties by defining an equivalence relation between a system under attack and the original unattacked system. That is, the system satisfies the robustness properties if executions of the attacked system are appropriately related to executions of the unattacked system.

We present two techniques for reasoning about the equivalence relation and thus proving the relational properties for a system.
One proof technique decomposes large proof obligations to smaller proof obligations. The other proof technique adapts the \emph{self-composition technique} from the literature on secure information-flow, allowing us to reduce reasoning about the equivalence of two systems to reasoning about properties of a single system. This technique allows us to reuse existing tools for reasoning about properties of hybrid programs, but is challenging due to the combination of discrete dynamics, continuous dynamics, and nondeterminism.


To validate the usefulness of our relational properties and proof techniques, we present three case studies motivated by real design flaws in existing cyber-physical systems. 


\end{abstract}


{\let\thefootnote\relax\footnote{{This is an extended version of the paper with the same title that appeared in the 2021 Computer Security Foundations Symposium. This version includes a proof of Theorem~\ref{theorem:soundness-composition} in Appendix~\ref{appendix:sound-proof}.}}

\setlength{\abovedisplayskip}{3pt}
\setlength{\belowdisplayskip}{2pt}
\setlength{\belowcaptionskip}{-10pt}

\section{Introduction}

Cyber-physical systems, which consist of both physical and cyber components, are often safety and security critical\cite{alur2011formal,bresolin2015formal,jeannin2015formally, mitsch2017formal}.
Designing secure cyber-physical systems is difficult because adversaries benefit from a broad attack surface that includes both software controllers and physical components.
Sensor attacks often allow an adversary to directly control the system under attack. 
For example, 
Cao et al. demonstrate how to manipulate an autonomous vehicle's distance measurements by shining a laser into its Light Detection and Ranging (LiDAR) sensors \cite{cao2019adversarial},
Humphreys et al. demonstrate how spoofing Global Positioning System (GPS) signals may allow an attacker to force a yacht autopilot to deviate from a designated course \cite{Yachtattack},
and Davidson et al. demonstrate a GPS-based hijacking attack on unmanned aircraft \cite{davidson2016controlling}.
The breadth of the cyber-physical attack surface affords adversaries a range of attack modalities even when hijacking control is not possible.
For example, Son et al. demonstrate how to crash a quadcopter using a magnetic attack on a quadcopter's gyroscopic sensors \cite{son2015rocking}.

Testing-based approaches are insufficient to guarantee the safety of a cyber-physical system, even when the system is not under attack. In a 2016 study on autonomous vehicles, Kalra et al. conclude that a self-driving fleet would need to drive hundreds of millions or sometimes hundreds of billions of miles to provide a purely testing-based reliability case \cite{kalra2016driving}. Driving these miles in a representative set of road conditions would take tens or hundreds of years depending on the size of the test fleet.
The intractability of testing-based approaches is also confirmed by incompleteness results \cite{platzer2012complete}.
Establishing security 
is even more difficult than establishing safety.

The importance and difficulty of ensuring the safety of cyber-physical systems motivate a growing body of work on formal verification for embedded and hybrid systems \cite{alur2015principles, larsen2009verification, lee2016introduction, tabuada2009verification, tiwari2011logic}.
However, relatively little work considers formal verification of such systems in the presence of \attackcategory. Some recent work emphasize timing aspects of sensor-related attacks \cite{lanotte2017formal, lanotte2020formal}; however, the work model the system's dynamics as a deterministic discrete time dynamical system, whereas most cyber-physical systems are best modeled with a nondeterministic combination of discrete and continuous dynamics.

It is important for cyber-physical system designers to understand
whether a compromised sensor can result in undesired behavior, such as
violating a safety property or corrupting a critical state.
For example, the designer of an adaptive cruise control system might want to verify that the car's minimum following distance is not affected by a compromised GPS sensor.

Understanding the impact of compromised sensors requires us to
reason about \emph{relational properties} \cite{clarkson2010hyperproperties}, that is, the
relationship between executions of the original uncompromised system
and executions of the system where some of the sensors have been
compromised. Relational properties are often harder to reason about than functional properties, as they require reasoning simultaneously about multiple executions. And there is less tool support for formal verification of relational properties, compared to functional properties. 


\smallskip

In this work, we define and explore two relational properties that
characterize the robustness of cyber-physical systems under sensor
attacks. Our threat model assumes a powerful attacker that may compromise a subset of sensors and arbitrarily manipulate those sensors' values. We do not model or discover the mechanisms by which an attacker manipulates sensor values; we simply assume they are able to do so.

Our first relational property is \emph{robustness of safety}, which
intuitively holds when the compromised sensors are unable to affect
whether a given safety property holds in the attacked system. Note
that this is not the same as requiring that the attacked system
satisfies the safety property. Indeed it may be beyond current
verification techniques to determine whether the safety property holds
in the uncompromised system, let alone the compromised
system. Nonetheless, even in such cases it can be possible to verify
that compromised sensors do not affect the safety property. 
Robustness of safety implies that if the uncompromised system
satisfies the safety property then the compromised system will
too. Reasoning about robustness of safety separates reasoning about
the implications of sensor attacks from reasoning directly about
functional properties.

Our second relational property is \emph{\sndrobustprop}, which
requires that high-integrity parts of a system cannot be influenced by
the attacker. For example, returning to our autonomous vehicle
example, parts of the system pertaining to steering and braking should
be regarded as high integrity and independent from low-integrity
sensors such as the interior thermometer. Robustness of high-integrity
state is similar to noninterference \cite{Sabelfeld2003,
  goguen1982security}, which requires that low-integrity inputs can
not influence 
high-integrity outputs.


We work within the formalism of \emph{hybrid programs}
\cite{platzer2008differential,Platzer18book,platzer2012complete} and
their implementation in the theorem prover KeYmaera~X
\cite{fulton2015keymaera}.  Hybrid programs model cyber-physical
systems as hybrid-time dynamical systems, with the discrete time
component of the system modeling software components and the
continuous time component of the system modeling physical phenomenon.

%

%
To define our two relational properties, we introduce the \emph{\Hequivalence} relation over hybrid programs, where $\Hsymbol$ is a set of variables. Intuitively, two hybrid programs are \Hequivalent if they agree on the values of all variables in $\Hsymbol$ at appropriate times. In particular, we define our two relational properties as \Hequivalence between the original system and the compromised system (for suitable sets of variables $\Hsymbol$).

We introduce two sound and tractable techniques to reason about \Hequivalence (and thus to prove that robustness of safety and robustness of high-integrity state hold). 
The first technique decomposes reasoning about $\Hequivalence$ of two large programs to reasoning about $\Hequivalence$ of their subprograms.
The second technique reduces reasoning about $\Hequivalence$ of two programs $A$ and $B$ to reasoning about safety  properties of a single program that represents both $A$ and $B$. This reduction allows us to prove relational properties using KeYmaera~X, an existing theorem prover for hybrid programs that does not directly support relational reasoning.
This technique is inspired by 
the self-composition technique \cite{barthe2004secure} used to prove noninterference in imperative and deterministic programs. A key challenge we faced in adapting the self-composition technique for hybrid programs is 
reasoning about nondeterminism and physical dynamics, and in particular, ensuring that certain nondeterministic choices are resolved the same in both executions.

The main contributions of this paper are the following:
\begin{enumerate}[label={\arabic*.}, noitemsep]
\item We introduce a threat model of \attackcategory in the context of hybrid programs that model cyber-physical systems. We show that these sensor attacks can be formalized in terms of syntactic manipulations of hybrid programs.
  We introduce robustness of safety and \sndrobustprop, two relational properties that express security guarantees in the presence of sensor attacks. (Section~\ref{sec:modeling-attacks})

\item We introduce \Hequivalence, an equivalence relation over hybrid programs, and express our relational properties in terms of \Hequivalence. (Section~\ref{sec:equivalence.relation})
  
\item We present two techniques for reasoning about \Hequivalence and prove their soundness. (Section~\ref{sec:prove-eq})
  
\item We validate the approach developed throughout the paper through three case studies of non-trivial cyber-physical systems: an anti-lock braking system, the Maneuvering Characteristics Augmentation System (MCAS) of the Boeing 737-MAX, and an autonomous vehicle with a shared communication bus. (Section~\ref{sec:casestudies})
\end{enumerate}

We introduce some background about hybrid programs in Section~\ref{sec:background}.  Section~\ref{sec:relatedwork} discusses related work. 

This is an extended version of the paper with the same title that appeared in the 2021 Computer Security Foundations Symposium. The main addition of this paper is the proof of Theorem~\ref{theorem:soundness-composition} in Appendix~\ref{appendix:sound-proof}.



\section{Background}  \label{sec:background}

\renewcommand{\arraystretch}{1.1} 

\begin{figure}[htb]
  \small\centering
  \begin{tabular}{p{0.15\linewidth}p{0.7\linewidth}}
\multicolumn{2}{l}{\textbf{Real-valued terms} $\theta$}  \\
    $ x $ & Real-valued program variable  \\
    $ c $ & Constant  \\
    $ \theta_1 \oplus \theta_2 $ & Computation on terms $\oplus \in \{ +, \times \}$  \\
\multicolumn{2}{l}{\textbf{Hybrid Program} $\alpha$, $\beta$, $\programwoattack$}  \\
    $ x:=\theta $ &  Deterministic assignment of real arithmetic term $\theta$ to variable $x$ \\
    $ x:=* $     &  Nondeterministic assignment to variable $x$ \\
    $ x'=\theta\& \evolconstraint $ &  Continuous evolution along the differential equation system $x'=\theta$ for an arbitrary real duration within the region described by formula $\evolconstraint$ \\
    $ ?\phi $ &  Test if formula $\phi$ is true at the current state  \\
    $\alpha ; \beta$ & Sequential composition of $\alpha$ and           $\beta$                                    \\
    $ \alpha \cup \beta $ &  Nondeterministic choice between $\alpha$ and 
                        $\beta$ \\
    $ \alpha^* $ & Nondeterministic repetition, repeating $\alpha$
                   zero or more times \\
    
    \multicolumn{2}{l}{\textbf{Differential Dynamic Logic} $\phi, \psi$} \\
      $ \theta_1 \sim \theta_2 $ & Comparison between real arithmetic terms ($\sim \in \{ <, \leq, =, >, \geq \}$)  \\
$\neg \phi$  & Negation  \\ 
$ \phi \land \psi $   &  Conjunction\\
$ \phi \lor \psi $  &   Disjunction\\
$ \phi \rightarrow \psi $  & Implication\\
$ \forall x.~ \phi $  & Universal quantification\\
$ \exists x.~ \phi $  & Existential quantification\\
    $ \HPbox{\alpha}\phi $  &   Program necessity (true if $\phi$ is true after each possible execution of hybrid program $\alpha$) \\    
    \end{tabular}
    \caption{Syntax of hybrid programs and $\differentiallogic$}
    \label{fig:syntax}
\end{figure}

\emph{Hybrid programs} \cite{Platzer18book} are a formalism for
modeling \emph{cyber-physical systems}, i.e., systems that have both
continuous and discrete dynamic behaviors. Hybrid programs can express
continuous evolution (as differential equations) as well as discrete
transitions.



Figure~\ref{fig:syntax} gives the syntax for hybrid programs. Variables are real-valued and can be deterministically assigned ($x := \theta$, where $\theta$ is a real-valued arithmetic term) or nondeterministically assigned ($x:=*$).
Hybrid program $x'=\theta\& \evolconstraint$ expresses the continuous evolution of variables: given the current value of variable $x$, the system follows the differential equation $x' = \theta$ for some (nondeterministically chosen) amount of time so long as the formula $\evolconstraint$, the \emph{evolution domain constraint}, holds for all of that time. Note that $x$ can be a vector of variables and then $\theta$ is a vector of terms of the same dimension.  

Hybrid programs also include the operations of Kleene algebra with tests \cite{kozen1997kleene}: sequential composition, nondeterministic choice, nondeterministic repetition, and testing whether a formula holds. 
Hybrid programs are models of systems and typically over-approximate the possible behaviors of a system.

\emph{Differential dynamic logic} ($\differentiallogic$) \cite{platzer2008differential, Platzer18book, 
platzer2017complete} is the dynamic logic~\cite{dynamiclogic} of hybrid programs. 
Figure~\ref{fig:syntax} also gives the syntax for 
$\differentiallogic$ formulas. 
In addition to the standard logical connectives of first-order logic, $\differentiallogic$ includes primitive propositions that allow comparisons of real-valued terms (which may include derivatives) and \emph{program necessity} $[\alpha]\phi$, which holds in a state if and only if after any possible execution of hybrid program $\alpha$, formula $\phi$ holds. 

The semantics of $\differentiallogic$ \cite{platzer2008differential, platzer2017complete} is a Kripke semantics in which the Kripke model's worlds
are the states of the system. Let $\realSet$ denote the set of
real numbers and $\allvariableSet$ denote the set of variables. A state is a
map $\traceState$ : $\allvariableSet$ $\mapsto$ $\realSet$ assigning a real value $\traceState(x)$ to each variable $x \in \allvariableSet$. The set of all states is denoted by $\stateSet$. The semantics of hybrid programs and $\differentiallogic$ are shown in Figure~\ref{fig:hpdlsem}. We write $\HPdlsem{\traceState}{\phi}$ if formula $\phi$ is true at state $\traceState$. The real value of term $\theta$ at state $\traceState$ is denoted $\HPtermsem{\traceState}{\theta}$.
The semantics of a hybrid program $\program$ is expressed as a transition relation
$\HPtransition{\program}$ between states.
If $\HPtransitionPair{\traceState}{\traceStateprime}$ $\in$ $\HPtransition{\program}$ then there is an execution of  $\program$ that starts in state $\traceState$ and ends in state $\traceStateprime$.



\begin{figure}[t]
  \small\centering
  \begin{tabular}{r@{\;\;}c@{\;\;}l}
\multicolumn{3}{l}{\textbf{Term semantics}} \\
    $ \traceState\HPtransition{x} $ & =& $\traceState(x)$  \\
    $ \traceState\HPtransition{c} $ & = &c  \\
    $ \traceState\HPtransition{\theta_1 \oplus \theta_2} $ & =& $\traceState\HPtransition{\theta_1} \oplus \traceState\HPtransition{\theta_2}$ for  $\oplus \in \{ +, \times \}$  \\
\multicolumn{3}{l}{\textbf{Program semantics}} \\
    $\HPtransition{\HPassignS{x}{\theta}}$ & = &$\{\HPtransitionPair{\traceState}{\traceStateprime} ~|~ \traceStateprime(x) = \traceState\HPtransition{\theta}$ and for all other\\&&\qquad variables $z\not=x$, $\traceStateprime(z) = \traceState(z)\}$  \\
    
    $\HPtransition{\HPassignS{x}{*}}$ & = &$\{\HPtransitionPair{\traceState}{\traceStateprime} ~|~ \traceStateprime(z) = \traceState(z)$ for all variables $z\not=x \}$  \\
  
    $\HPtransition{?\phi}$ & =& $\{\HPtransitionPair{\traceState}{\traceState} ~|~ \traceState \models \phi \}$ \\
    
    $\HPtransition{x'=\theta$ $\& \evolconstraint}$ & = &$\{\HPtransitionPair{\traceState}{\traceStateprime} ~|$ iff exists solution $\varphi:[0,r] \mapsto \stateSet$ of \\&&\qquad $x'=\theta$ with $\varphi(0)=\traceState$ and \\&&\qquad$\varphi(r)=\traceStateprime$, and $\HPdlsem{\varphi(t)}{\evolconstraint}$ for all $t \in [0,r] \}$  \\

    $\HPtransition{\alpha \cup \beta}$ & = &$\HPtransition{\alpha} \cup \HPtransition{\beta}$ \\

   $\HPtransition{\alpha ; \beta}$ & = &$\{\HPtransitionPair{\traceState}{\traceStateprime} ~|~ \exists \mu, \HPtransitionPair{\traceState}{\mu} \in \HPtransition{\alpha} \text{~and~}
     \HPtransitionPair{\mu}{\traceStateprime} \in \HPtransition{\beta} \}$  \\
                                     
    $\HPtransition{\alpha^*}$ & =& $\HPtransition{\alpha}^*$ the transitive, reflexive closure of $\HPtransition{\alpha}$ \\
    
\multicolumn{3}{l}{\textbf{Formula semantics}}  \\
    
    $\HPdlsem{\traceState}{\theta_1 \sim \theta_2}$ & iff &$\HPtermsem{\traceState}{\theta_1} \sim \HPtermsem{\traceState}{\theta_2}$ for  $\sim \in \{ =, \leq, <, \geq, > \}$ \\
    
    $\HPdlsem{\traceState}{\phi \land \psi}$ & iff &$\HPdlsem{\traceState}{\phi}$ $\land$ $\HPdlsem{\traceState}{\psi}$, similar for $\{ \neg, \lor, \rightarrow, \leftrightarrow \}$ \\

    $\HPdlsem{\traceState}{\forall x.\phi}$ & iff& $\HPdlsem{\traceStateprime}{\phi}$ for all states $\traceStateprime$ that agree with $\traceState$ \\&&\qquad except for the value of $x$ \\

    $\HPdlsem{\traceState}{\exists x.\phi}$ & iff& $\HPdlsem{\traceStateprime}{\phi}$ for some state $\traceStateprime$ that agrees with $\traceState$ \\&&\qquad except for the value of $x$ \\

    $\HPdlsem{\traceState}{ \HPbox{\alpha} \phi}$ & iff& $\HPdlsem{\traceStateprime}{\phi}$ for all state $\traceStateprime$ with $\HPtransitionPair{\traceState}{\traceStateprime}$ $\in$ $\HPtransition{\alpha}$
  \end{tabular}
    \caption{Semantics of hybrid programs and $\differentiallogic$}
    \label{fig:hpdlsem}
\end{figure}

\label{hybridprogramexample}
We are often interested in partial correctness formulas of the form
$\phi \rightarrow [\alpha]\psi$: if $\phi$ is true then $\psi$ holds
after any possible execution of $\alpha$.  The hybrid program $\alpha$
often has the form \HPgeneralform, where
\inlineCode{ctrl} models atomic actions of the control system and does not contain continous parts (i.e., differential equations); and
\inlineCode{plant} models evolution of the physical environment and has the form of $x'=\theta$ $\& \evolconstraint$. 
That is, the system is modeled as
unbounded repetitions of a controller action
followed by an update to the physical environment.

Consider, as an example, an autonomous vehicle that needs to stop
before hitting an obstacle.\footnote{Platzer introduces this autonomous vehicle example \cite{Platzer18book}.}
For simplicity, we model the vehicle in just one dimension. 
Figure~\ref{fig:eg-vehicle} shows a \emph{\systemmodel} (hybrid program model) of such an autonomous vehicle.
\footnote{Syntax of hybrid programs used in this paper is similar to the syntax used in KeYmaera X, but revised for better presentation.} 
Let \inlineCode{$d$} be the vehicle's distance from the obstacle. 
The \emph{safety condition} that we would like to enforce
(\inlineCode{$\phi_\mathit{post}$}) is that \inlineCode{$d$} is positive.
Let \inlineCode{$v$} be the vehicle's velocity towards the obstacle in meters per second (m/s) and let \inlineCode{$a$} be the vehicle's acceleration (m/s${}^2$). 
Let \inlineCode{$t$} be the time elapsed since the controller was last invoked. 
The hybrid program \inlineCode{plant} describes how
the physical environment evolves over time interval \inlineCode{$\epsilon$}: 
distance changes according to \inlineCode{$-v$} (i.e., \inlineCode{$d^\prime = -v$}), 
velocity changes according to the acceleration (i.e., \inlineCode{$v^\prime = a$}), 
and time passes at a constant rate (i.e., \inlineCode{$t^\prime=1$}). 
The differential equations evolve only within the time interval \inlineCode{$t \le \epsilon$} and if \inlineCode{$v$} is non-negative 
(i.e., \inlineCode{$v \geq 0$}).

The hybrid program \inlineCode{ctrl} models the vehicle's controller.  The
vehicle can either accelerate at \inlineCode{$A$} m/$s^2$ or brake at \inlineCode{$-B$} m/$s^2$. For the purposes of the model, the
controller chooses nondeterministically between these options. Hybrid
programs \inlineCode{accel} and \inlineCode{brake} express
the controller accelerating or braking (i.e., setting \inlineCode{$a$} to \inlineCode{$A$} or \inlineCode{$-B$} respectively). The controller can accelerate only
if condition \inlineCode{$\psi$} is true, which captures that the
vehicle can accelerate for the next \inlineCode{$\epsilon$} seconds only if doing
so would still allow it to brake in time to avoid 
the obstacle.


\newcommand\cmts{5cm}

\newcommand\codepos{6.8cm}
\newcommand\cmtpos{6.8cm}
\newcommand\linespacing{1.1}

\begin{figure}
\begin{lstlisting}[basicstyle={\linespread{\linespacing}\ttfamily\scriptsize}]    
Definitions. `\javacommentalign{\cmts}{cannot change over time}`
 R $\epsilon$. `\javacommentalign{\cmtpos}{time limit for control}`
 R $A$. `\javacommentalign{\cmtpos}{acceleration rate}`
 R $B$. `\javacommentalign{\cmtpos}{braking rate}`
 B $\phi_\mathit{pre}$ `\codealign{\codepos}{\HPdef $A \geq 0 \land B \geq 0 \land 2Bd > v^2$}`
 B $\phi_\mathit{post}$ `\codealign{\codepos}{\HPdef $d > 0$}`
 B $\psi$ `\codealign{\codepos}{\HPdef $2Bd > v^2 + (A+B)(A\epsilon^2 + 2v\epsilon)$}`
 HP accel `\codealign{\codepos}{\HPdef $?\psi; a:= A$}`
 HP brake `\codealign{\codepos}{\HPdef $a := -B$}`
 HP ctrl `\codealign{\codepos}{\HPdef ((accel $\cup$ brake); $t:=0$)}`
 HP plant `\codealign{\codepos}{\HPdef $d^\prime = -v, v^\prime = a, t^\prime = 1 \; \& \; (v \geq 0 \land t \leq \epsilon)$}`
ProgramVariables. `\javacommentalign{\cmts}{may change over time}`
 R $t$. `\javacommentalign{\cmtpos}{clock variable}`
 R $d$. `\javacommentalign{\cmtpos}{distance to obstacle}`
 R $v$. `\javacommentalign{\cmtpos}{vehicle velocity}`
 R $a$. `\javacommentalign{\cmtpos}{acceleration of the vehicle}`
Problem. `\javacommentalign{\cmts}{dL formula to be proven}`
 $\phi_\mathit{pre}$ $\rightarrow$ [(ctrl; plant)$^*$]$\phi_\mathit{post}$
\end{lstlisting}
\caption{\systemmodel of an autonomous vehicle}\label{fig:eg-vehicle}
\end{figure}


The formula to be verified is presented on the last line of the \systemmodel.
Given an appropriate precondition \inlineCode{$\phi_\mathit{pre}$}, the axioms and proof rules $\differentiallogic$ can be used to prove that the safety condition \inlineCode{$\phi_\mathit{post}$} holds.
The tactic-based theorem prover KeYmaera X \cite{fulton2015keymaera} provides tool support for automating the construction of these proofs.


To present some of our definitions, we need to refer to the variables that occur in a hybrid program\cite{Platzer18book,platzer2017complete}.
The \emph{free variables} of hybrid program $\programwoattack$, denoted $\FV{\programwoattack}$, is the variables that may potentially be read by $\programwoattack$. Values of $\FV{P}$ won't be modified during executions of program $\programwoattack$. The \emph{bound variables} of program $\programwoattack$, denoted $\BV{\programwoattack}$, is the set of variables that may potentially be written to by $\programwoattack$.\footnote{We follow the naming convention of related work on hybrid programs by using the names of free variables and bound variables\cite{platzer2017complete}.}
We write $\variableSet{\programwoattack}$ for the set of all variables of  $\programwoattack$, and have $\variableSet{\programwoattack} = \BV{\programwoattack} \cup \FV{\programwoattack}$.
For example, let $\programwoattack$ be the hybrid program modeling an autonomous vehicle with sensors shown in Figure~\ref{fig:eg-vehicle}, then $\FV{\programwoattack}$ = $\{ A, B, \epsilon, v, d \}$, $\BV{\programwoattack}$ = $\{ t, v, d, a, t^\prime, v^\prime, d^\prime\}$, and $\variableSet{\programwoattack}$ = $\{A, B, \epsilon, t, v, d, a, t^\prime, v^\prime, d^\prime \}$.
Formal definitions of $\BV{\programwoattack}$, $\FV{\programwoattack}$, and $\variableSet{\programwoattack}$ are included in Appendix~\ref{appendix:definitions}.   



\section{Modeling Sensor Attacks} \label{sec:modeling-attacks}

In this section, we explain how we model the \attackcategory in hybrid programs. In particular, we introduce how sensor readings are modeled and describe our threat model.   

\subsection{Modeling Sensor Readings} \label{sec:sensor-modeling}

Hybrid programs typically conflate the values of variables in the physical model and the values ultimately perceived by the sensor. 
For example, in Figure~\ref{fig:eg-vehicle}, the hybrid program contains a single continuous variable \inlineCode{$v$} that represents the value measured by a sensor; the model does not separate the model's representation of the value of \inlineCode{$v$} in the physical model from the software component's representation of \inlineCode{$v$}.
Therefore, our analysis begins with a hybrid program $\hybridprogram_\mathit{orig}$ in which sensor reads are not explicitly modeled.
We construct a program $\hybridprogram$ that is equivalent to $\hybridprogram_\mathit{orig}$ but separately represents sensor reads and requires that variables holding sensor reads are equal to the underlying sensor's value. 
For example, \inlineCode{$v_p$} may represent the actual physical velocity of a vehicle and it changes according to laws of physics, and \inlineCode{$v_s$} may represent the variable in the controller into which the sensor's value is read.
In model $\hybridprogram$ we have the constraint \inlineCode{$v_s=v_p$}.
From $\hybridprogram$ we can derive additional models that allow
sensed values to differ from actual physical values. For example,
a model that represents the compromise of the velocity sensor would be
identical to $\hybridprogram$ except that the constraint \inlineCode{$v_s=v_p$} is removed, allowing \inlineCode{$v_s$} to take arbitrary
values. 
Similar modifications to $\hybridprogram$ can represent the compromise of
other sensors, or of multiple sensors at the same time.

As an example, Figure~\ref{fig:eg-vehicle-sensing} shows a \systemmodel of an autonomous vehicle introduced in Figure~\ref{fig:eg-vehicle} whose hybrid program separates physical and sensed values: \inlineCode{$v_p$} and \inlineCode{$d_p$} are physical values of velocity and distance, while \inlineCode{$v_s$} and \inlineCode{$d_s$} are the corresponding sensed values. Note that the \inlineCode{ctrl} program sets the sensed values equal to the physical values (line~\ref{line:eg-sensing-contraint}).

\begin{figure}
\begin{lstlisting}[basicstyle={\linespread{\linespacing}\ttfamily\scriptsize}]    
Definitions. 
 R $\epsilon$. `\javacommentalign{\cmtpos}{time limit for control}`
 R $A$. `\javacommentalign{\cmtpos}{acceleration rate}`
 R $B$. `\javacommentalign{\cmtpos}{braking rate}`
 B $\phi_\mathit{pre}$ `\codealign{\codepos}{\HPdef $A \geq 0 \land B \geq 0 \land 2Bd_p > v_p^2$}`
 B $\phi_\mathit{post}$ `\codealign{\codepos}{\HPdef $d_p > 0$}`
 B $\psi$ `\codealign{\codepos}{\HPdef $2Bd_s > v_s^2 + (A+B)(A\epsilon^2 + 2v_s\epsilon)$}`
 HP accel `\codealign{\codepos}{\HPdef $?\psi; a:= A$}`
 HP brake `\codealign{\codepos}{\HPdef $a := -B$}`
 HP ctrl `\codealign{\codepos}{\HPdef $\HPassignment{v_s}{v_p}$ $\HPassignment{d_s}{d_p}$ (accel $\cup$ brake); $t:=0$ \label{line:eg-sensing-contraint}}` 
 HP plant `\codealign{\codepos}{\HPdef $d_p^\prime = -v_p, v_p^\prime = a, t^\prime = 1 \; \& \; (v_p \geq 0 \land t \leq \epsilon)$}`
ProgramVariables.
 R $t$. `\javacommentalign{\cmtpos}{clock variable}`
 R $d_p$. `\javacommentalign{\cmtpos}{distance to obstacle (physical)}`
 R $v_p$. `\javacommentalign{\cmtpos}{vehicle velocity (physical)}`
 R $d_s$. `\javacommentalign{\cmtpos}{distance to obstacle (sensed)}`
 R $v_s$. `\javacommentalign{\cmtpos}{vehicle velocity (sensed)}`
 R $a$. `\javacommentalign{\cmtpos}{acceleration of the vehicle}`
Problem.
 $\phi_\mathit{pre}$ $\rightarrow$ [(ctrl; plant)$^*$]$\phi_\mathit{post}$
\end{lstlisting}
\caption{\systemmodel of an autonomous vehicle with sensors}\label{fig:eg-vehicle-sensing}
\end{figure}

\subsection{Threat Model} \label{sec:safety-thread-model}

We allow attackers to arbitrarily change sensed values. We are not concerned with the physical mechanisms by which an attacker compromises a sensor. Instead, we model sensor attacks as assignments to variables that represent sensed values. Let $\programwoattack$ be a hybrid program, $S_A \subseteq \BV{\programwoattack}$ be a set of distinguished variables corresponding to sensors that may be vulnerable to attacks, the sensor attack on $\programwoattack$ is defined as follows:

\begin{definitions}[\attackname]
For a hybrid program $\programwoattack$ of the form \HPgeneralform and a set of variables $S_A$ $\subseteq$ $\BV{\programwoattack}$, the \emph{\attackname} on program $\programwoattack$, denoted $\attacked{\programwoattack}{S_A}$, is the program obtained from $\programwoattack$ by replacing all assignments to variable $v \in S_A$ with assignment $v := *$. 
\end{definitions}

For example, let  $\programwoattack$ be the hybrid program \inlineCode{(ctrl;plant)$^*$}  modeling an autonomous vehicle with separate physical and sensed values shown in Figure~\ref{fig:eg-vehicle-sensing}. If the velocity sensor \inlineCode{$v_s$} is under attack, program $\attacked{\programwoattack}{\{v_s\}}$ would be \inlineCode{(ctrl$^{\prime}$;plant)$^*$} where \inlineCode{ctrl$^\prime$} is the following:
\[
\highlight{\HPassignment{v_s}{*}} \; \HPassignment{d_s}{d_p} \; (\texttt{accel}  \; \cup \; \texttt{brake}); t:=0.
\]
Note that with such a threat model, only the \inlineCode{ctrl} part of a program \HPgeneralform is modified by an attack, i.e., \attacked{\inlineCode{(ctrl; plant)}}{$S_A$} = (\attacked{\inlineCode{ctrl}}{$S_A$}); \inlineCode{plant}. Intuitively, it means a sensor attack does not \emph{directly} affect the physical dynamics with which the system interacts.

\subsection{Robustness to Sensor Attacks} \label{sec::properties}

%
We explore the impact of an \attackname by studying two relational properties that characterize the robustness of the system to the attack: (1) whether a \attackname affects the safety of the system and (2) whether a \attackname affects the system's high-integrity state.

\medskip

\paragraph{Robust Safety} Safety is critical in many cyber-physical systems, e.g., a vehicle should not collide with obstacles and pedestrians.
We first present the definitions of safety and our relational property robust safety, and then show an example.

\begin{definitions}[Safety] \label{def:safety} 
A hybrid program $\hybridprogram$ of the form \HPgeneralform is \emph{safe for $\phi_{post}$ assuming $\phi_{pre}$}, denoted $\safety{\hybridprogram}{\phi_{pre}}{\phi_{post}}$,
if the formula $\HPproperty{\phi_{pre}}{[\hybridprogram]}{\phi_{post}}$ holds.
\end{definitions}

This definition says $\hybridprogram$ is safe if for any execution of $\hybridprogram$ whose starting state satisfies $\phi_{pre}$, its ending state satisfies safety condition $\phi_{post}$. 

A system is \emph{robustly safe} if compromise of sensors  $S_A$ does not affect whether the system is safe. Note that robust safety does not require that the attacked system is safe; instead it requires that \emph{if} the original system is safe, then the attacked system is also safe. The distinction is important: it allows us to separate the task of reasoning about safety from the task of reasoning about sensor attacks. Indeed, as we will see in a case study in Section~\ref{sec:casestudies}, it is possible to prove robust safety even when it is beyond current techniques to prove safety.

\begin{definitions}[Robust safety] \label{def:robustness-safety}
  For a hybrid program $\programwoattack$ of the form \HPgeneralform and a set of variables $S_A$ $\subseteq$ $\BV{\programwoattack}$, $\programwoattack$ is \emph{robustly safe for $\phi_{post}$ assuming $\phi_{pre}$ under the \attackname}, denoted $\robustsafety{\programwoattack}{\phi_{pre}}{\phi_{post}}{S_A}$, if
  $\safety{\programwoattack}{\phi_{pre}}{\phi_{post}}$ implies $\safety{\attacked{\programwoattack}{S_A}}{\phi_{pre}}{\phi_{post}}$.
\end{definitions}


\begin{figure}
\begin{lstlisting}[basicstyle={\linespread{\linespacing}\ttfamily\scriptsize}]    
 ...
 HP voting `\codealign{\codepos}{\HPdef \HPassignment{$v_{s_1}$}{$v_p$}  \HPassignment{$v_{s_2}$}{$v_p$} \HPassignment{$v_{s_3}$}{$v_p$}\label{line:vehicle-robustness-safety-3sensors}}`
 `\codealign{\codepos}{~~~(\phantom{$\cup$~}(\HPguard{$v_{s_1}$ = $v_{s_2}$} \HPassignS{$v_s$}{$v_{s_1}$})} \label{line:vehicle-robustness-safety-voting1}`
 `\codealign{\codepos}{~~~~$\cup$~(\HPguard{$v_{s_1}$ = $v_{s_3}$} \HPassignS{$v_s$}{$v_{s_1}$})} \label{line:vehicle-robustness-safety-voting2}`
 `\codealign{\codepos}{~~~~$\cup$~(\HPguard{$v_{s_2}$ = $v_{s_3}$} \HPassignS{$v_s$}{$v_{s_2}$})~)} \label{line:vehicle-robustness-safety-voting3}`
 HP ctrl `\codealign{\codepos}{\HPdef voting; $\HPassignment{d_s}{d_p}$ (accel $\cup$ brake); $t:=0$}`
 ...
\end{lstlisting}
\caption{\systemmodel of an autonomous vehicle with sensor voting} \label{fig:eg-vehicle-robustness-safety}
\end{figure}


For example, let $\programwoattack$ be the hybrid program modeling an autonomous vehicle with sensors shown in Figure~\ref{fig:eg-vehicle-sensing}. $\programwoattack$ is safe for $\phi_\mathit{post}$ assuming $\phi_\mathit{pre}$ (i.e., $\safety{\programwoattack}{\phi_{pre}}{\phi_{post}}$). However, $\programwoattack$ is not robustly safe for $\phi_\mathit{post}$ assuming $\phi_\mathit{pre}$ under \attackname where $S_A$ is $\{ v_{s}\}$, since $\safety{\attacked{\programwoattack}{S_A}}{\phi_{pre}}{\phi_{post}}$ doesn't hold.


The system can be modified so that it does satisfy robust safety. For example,
we can modify the system to use three velocity sensors (perhaps measuring velocity by different mechanisms) and use a voting scheme to determine the current velocity.
Figure~\ref{fig:eg-vehicle-robustness-safety} shows a model of such a modified system.
The physical velocity $v_p$ is sensed by three sensors (line~\ref{line:vehicle-robustness-safety-3sensors}), and voting performed to determine the final reading $v_s$ (lines~\ref{line:vehicle-robustness-safety-voting1}--\ref{line:vehicle-robustness-safety-voting3}). The contents elided in Figure~\ref{fig:eg-vehicle-robustness-safety} are the same as Figure~\ref{fig:eg-vehicle-sensing}. 




Let $\programwoattack$ be the hybrid program modeling an autonomous vehicle with duplicated sensors shown in Figure~\ref{fig:eg-vehicle-robustness-safety}. For any set $S_A \in \{ \{v_{s_1}\}, \{v_{s_2}\}, \{v_{s_3}\} \}$,
program  $\programwoattack$ is robustly safe under $S_A$-sensor attack, i.e., $P$ is robustly safe if at most one of the velocity sensors is compromised.
Intuitively, this is because  $v_s = v_p$ holds after running program \inlineCode{voting}, even if up to one of the velocity sensors is compromised.
A systematic approach for proving robustness safety is presented in Section~\ref{sec:prove-eq}.

\medskip

\paragraph{Robustness of High-Integrity State}


Sensors that may be compromised are low integrity: the sensed values might be under the control of the attacker. By contrast, parts of the system state might be deemed to be high integrity: their values are critical to the correct and secure operation of the system. Low-integrity sensor readings should not be able to affect a system's high-integrity state.  
For example, an attacker with access to a car's interior temperature sensor should not be able to affect the control of the car's velocity.

We can state this requirement as a relational property: we say the high-integrity state is robust if, for any execution of the system with its low-integrity sensors compromised, there is an execution of the non-compromised system that can achieve the same values on all high-integrity variables. We delay formal definition of  robustness of high-integrity state to Section~\ref{sec:equivalence.relation}.

\begin{figure}
\begin{lstlisting}[basicstyle={\linespread{\linespacing}\ttfamily\scriptsize}]    
 ...
 R $T$. `\javacommentalign{\cmtpos}{target temperature}`
 HP ctrl$_{t}$ `\codealign{\codepos}{\HPdef \HPassignment{$temp_s$}{$temp_p$}}`
 `\codealign{\codepos}{~~~(\phantom{$\cup$~}(\HPguard{$temp_s>T$} \HPassignS{$thermo$}{-1}) \label{line:eg-vehicle-temp-contrl1}}`
 `\codealign{\codepos}{~~~~$\cup$~(\HPguard{$temp_s<T$} \HPassignS{$thermo$}{1})}`
 `\codealign{\codepos}{~~~~$\cup$~($?temp_s=T$)~) \label{line:eg-vehicle-temp-contrl2}} `
 HP ctrl `\codealign{\codepos}{\HPdef ctrl$_t$;}`
 `\codealign{\codepos}{~~~~\HPassignment{$v_s$}{$v_p$} \HPassignment{$d_s$}{$d_p$} (accel $\cup$ brake); $t:=0$}`
 HP plant `\codealign{\codepos}{\HPdef $d_p^\prime = -v_p, v_p^\prime = a, \; temp_{p}^{\prime} = thermo; \;  t^\prime = 1 \;$}`
 `\codealign{\codepos}{~~~~$\& \; (v_p \geq 0 \land t \leq \epsilon)$}`
ProgramVariables.
 R $temp_s$. `\javacommentalign{\cmtpos}{interior temperature (sensed)}`
 R $temp_p$. `\javacommentalign{\cmtpos}{interior temperature (physical)}`
 R $thermo$. `\javacommentalign{\cmtpos}{thermostat command}`
 ...
\end{lstlisting}
\caption{\systemmodel of an autonomous vehicle with interior temperature control} \label{fig:eg-vehicle-robustness-guard}
\end{figure}

Let's consider an example. Figure~\ref{fig:eg-vehicle-robustness-guard} presents a \systemmodel of an autonomous vehicle with sensors shown in Figure~\ref{fig:eg-vehicle-sensing} but added with interior temperature control (elided contents in Figure~\ref{fig:eg-vehicle-robustness-guard} are the same as Figure~\ref{fig:eg-vehicle-sensing}). 
The vehicle has sensor readings of interior temperature (\inlineCode{$temp_s$}). The physical temperature (\inlineCode{$temp_p$}) changes according to \inlineCode{$thermo$} that is set by \inlineCode{ctrl$_{t}$} after comparing \inlineCode{$temp_s$} with target temperature \inlineCode{$T$} (lines~\ref{line:eg-vehicle-temp-contrl1}--\ref{line:eg-vehicle-temp-contrl2}).
In this example, the temperature sensor is low-integrity and may be compromised.

A system designer may want to understand if such an attack can interfere with the vehicle's high-integrity state such as its velocity. Let $\programwoattack$ be the model of an autonomous vehicle with interior temperature control shown in Figure~\ref{fig:eg-vehicle-robustness-guard}.
Intuitively, its velocity (i.e., variable  $v_p$) is robust with respect to sensor \inlineCode{$temp_s$}: for any execution of \attacked{$\programwoattack$}{\{\inlineCode{$temp_s$}\}}, we have an execution of $\programwoattack$ that can produce the same values of $v_p$ at every control iteration. The system does satisfy robustness of high-integrity state, and we will prove it in Section~\ref{sec:prove-eq}.




\section{\Hequivalence} \label{sec:equivalence.relation}

This section introduces $\Hequivalence$, a notion of equivalence that allows us to reason about our relational properties.
%

\subsection{Equivalence of Hybrid Programs}

Intuitively, $\Hequivalence$ of two programs means that for every execution of one program, there exists an execution of the other program such that the two executions agree on set $\Hsymbol$ initially and at the end of every control loop iteration, where $\Hsymbol$ is a set of high-integrity variables.

The formal definition of  $\Hequivalence$ of programs builds on $\Hequivalence$ of program states. 


\begin{definitions}[$\Hequivalence$ of program states]  \label{def:equivlance-state} For states $\traceState_1, \traceState_2 \in \stateSet$ and a set of variables $\Hsymbol$, 
states $\traceState_1$ and $\traceState_2$ are $\Hequivalent$, denoted $\worldloweq{\traceState_1}{\traceState_2}{\Hsymbol}$,  if they agree on valuations of all variables in the set $\Hsymbol$; i.e., $ \forall x \in \Hsymbol, \traceState_1(x) = \traceState_2(x) $.
\end{definitions}
  


\begin{definitions}[$\Hequivalence$ of programs]
  \label{def:equivlance}For hybrid programs $\hybridprogram_1$ $=$ $\alpha^*$, $\hybridprogram_2$ $=$ $\beta^*$, and a set of variables $\Hsymbol$,
  $\hybridprogram_1$ and $\hybridprogram_2$ are \emph{\Hequivalent}, denoted $\eqprogram{\hybridprogram_1}{\hybridprogram_2}{\Hsymbol}$, if they satisfy the following:
  \begin{align*}
  \forall & n : \mathbb{N} \\
          & \forall \traceState_0, \traceState_1 \ldots \traceState_n : \stateSet \suchthat \forall i \in 0 ... (n-1), \\
          & ~\HPtransitionPair{\traceState_i}{\traceState_{i+1}} \in \HPtransition{\alpha}~(\text{respectively~} \HPtransition{\beta})\\
          & ~~\exists \traceStateprime_0, \traceStateprime_1 ... \traceStateprime_n : \stateSet \suchthat \forall j \in 0 ... (n-1), \\
          & ~~~ \HPtransitionPair{\traceStateprime_j}{\traceStateprime_{j+1}} \in \HPtransition{\beta}~(\text{respectively~} \HPtransition{\alpha}) \\
          & ~~~~ \text{and~} \forall k \in 0 ... n,~ \worldloweq{\traceState_k}{\traceStateprime_k}{\Hsymbol}
  \end{align*} 
\end{definitions}

In the definition, the number $n$ corresponds to an arbitrary number of loop iterations, and the last line indicates that the two executions agree on $\Hsymbol$ at the beginning and end of every loop iteration. The definition is symmetric. 


This definition can be readily adjusted for loop-free programs. 
\begin{definitions}[$\Hequivalence$ of loop-free programs]
  \label{def:equivlance-loop-free}For two loop-free hybrid programs $\alpha$ and $\beta$, and a set of variables $\Hsymbol$,
$\alpha$ and $\beta$ are \emph{\Hequivalent}, denoted $\eqprogram{\alpha}{\beta}{\Hsymbol}$, if they satisfy the following:
\begin{align*}
  \forall
  & \traceState_0, \traceState_1 : \stateSet \suchthat
  \HPtransitionPair{\traceState_0}{\traceState_1} \in \HPtransition{\alpha} ~(\text{respectively~}\HPtransition{\beta})  \\
  & \exists  \traceStateprime_0, \traceStateprime_1 : \stateSet \suchthat \\
  & ~\HPtransitionPair{\traceStateprime_0}{\traceStateprime_{1}} \in \HPtransition{\beta}  ~(\text{respectively~}\HPtransition{\alpha}) \land \worldloweq{\traceState_0}{\traceStateprime_0}{\Hsymbol} \land
    \worldloweq{\traceState_1}{\traceStateprime_1}{\Hsymbol} 
\end{align*}
\end{definitions}


Note that Definition~\ref{def:equivlance} is defined in lock-step, i.e., both loops iterate
exactly the same number of times \cite{pick2018exploiting}.
As pointed out by previous work \cite{shemer2019property}, a lock-step approach is sometimes not flexible enough to express and verify some properties, e.g., properties that may hold for two programs that execute for different numbers of iterations.
However, such a lock-step definition is reasonable in our setting.
According to the threat model, we are comparing a system with compromised sensors and a system with uncompromised sensors and so the attack should not affect the rate of a system's control (i.e., how frequently the system's control loop executes). 
Thus, the robustness of a system is correctly encoded by a lock-step definition, in which states of a system with and without compromised sensors are consistent after every loop iteration.
An additional benefit of this definition is that it is more tractable for verification, which we will explore in Section~\ref{sec:prove-eq}.

\subsection{Reasoning about Robustness using $\Hequivalence$}


The $\Hequivalence$ relation can be used to reason about our two relational properties.

\medskip


\paragraph{Reasoning about Robustness of Safety} 
Robustness of safety can be established by proving $\Hequivalence$ with the help of the following theorem, which states that if program $P$ is \Hequivalent to $\attacked{\programwoattack}{S_A}$ where $\Hsymbol$ is the free variables of formulas $\phi_{pre}$ and $\phi_{post}$, then $P$ is robustly safe for $\phi_{post}$ assuming $\phi_{pre}$ under the \attackname.
 
\begin{theorem}[$\Hequivalent$ programs are robustly safe] \label{lemma:robust-safety} For a hybrid program $\programwoattack$ of the form \HPgeneralform, a set of variables $S_A \subseteq \BV{\programwoattack}$, and formulas $\phi_{pre}$ and $\phi_{post}$, if $\eqprogram{\programwoattack}{\attacked{\programwoattack}{S_A}}{\FV{\phi_{pre} \land \phi_{post}}}$, then 
  \[  \robustsafety{\programwoattack}{\phi_{pre}}{\phi_{post}}{S_A} \]
\end{theorem}


A proof is in Appendix~\ref{appendix:proofs}. Intuitively, the theorem holds because if there were an execution of attacked program such that $\phi_{pre}$ held at the beginning but $\phi_{post}$ did not hold at the end of a loop, then there must be an execution of $\programwoattack$ where the same is true, contradicting the assumption that $\programwoattack$ is safe. 

Note that the converse of Theorem~\ref{lemma:robust-safety} does not hold, i.e., if $\robustsafety{\programwoattack}{\phi_{pre}}{\phi_{post}}{S_A}$, it is not always true that $\eqprogram{\programwoattack}{\attacked{\programwoattack}{S_A}}{\FV{\phi_{pre} \land \phi_{post}}}$.
For example, let $\programwoattack$ be the program $(\HPassignment{b}{1}~\HPassignS{a}{b})^*$, formula $\phi_{pre}$ be $a > 0$, $\phi_{post}$ be $b > 0$, and $S_A$ be $\{ a \}$. Then $\robustsafety{\programwoattack}{\phi_{pre}}{\phi_{post}}{S_A}$ holds, but $\eqprogram{\programwoattack}{\attacked{\programwoattack}{S_A}}{\{ a, b \}}$ does not hold since some executions of $\attacked{\programwoattack}{S_A}$ (i.e., $(\HPassignment{b}{1}~\HPassignS{a}{*})^*$) do not have a matching execution of $\programwoattack$.

Theorem~\ref{lemma:robust-safety} reduces proving robustness of safety to proving $\Hequivalence$, which can be achieved by the techniques introduced in Section~\ref{sec:prove-eq}.


\medskip

\paragraph{Reasoning about Robustness of High-Integrity State}
$\Hequivalence$ directly expresses robustness of high-integrity state by letting $\Hsymbol$ be the set of high-integrity variables. 
Therefore, proving robustness of high-integrity state is the \emph{same} as proving $\Hequivalence$ of the high-integrity state. The following definition 
makes this clear.


\begin{definitions}[Robustness of high-integrity state]
  For program $\programwoattack$ of the form \HPgeneralform and a set of variables $S_A$ $\subseteq$ $\BV{\programwoattack}$,  and set of variables $\Hsymbol$,
  $\programwoattack$ satisfies \emph{robustness of high-integrity state $\Hsymbol$ under the \attackname}
  if $\eqprogram{\programwoattack}{\attacked{\programwoattack}{S_A}}{\Hsymbol}$.
\end{definitions}







\section{Proving $\Hequivalence$} \label{sec:prove-eq}


We present two sound techniques for reasoning about $\Hequivalence$.

\subsection{Decomposition Approach} \label{sec:prove-eq-decomposition}

Our first approach proves $\Hequivalence$ of programs by \emph{decomposing} the proof obligation into simpler obligations for components of the programs. This relies on various compositional properties of $\Hequivalence$, stated here and proven in Appendix~\ref{appendix:proofs}.

\begin{theorem} \label{theorem:loop-free-properties}
For all loop-free hybrid programs $A$, $B$, $C$, $D$ and sets $\Hsymbol$ and $\Hsymbol'$ of variables, the following properties hold:
\begin{enumerate}[label={\arabic*.}, ref=\arabic*, itemsep=-1ex] 
  \item \label{theorem:eq-self} $\worldloweq{A}{A}{\Hsymbol}$;  
  \item \label{theorem:eq-subset} If $\Hsymbol \subseteq \Hsymbol'$ and  $\worldloweq{A}{B}{\Hsymbol'}$, then $\worldloweq{A}{B}{\Hsymbol}$; 
\item \label{theorem:eq-unmodified} If $\worldloweq{A}{B}{\Hsymbol}$ and $(\variableSet{A} \cup \variableSet{B}) ~\cap~ \Hsymbol' = \emptyset$, then $\worldloweq{A}{B}{\Hsymbol \cup \Hsymbol'}$;
\item \label{theorem:eq-decomposition} If $\FV{C} \cup \FV{D} \subseteq {\Hsymbol}$, $\worldloweq{A}{B}{\Hsymbol}$, and $\worldloweq{C}{D}{\Hsymbol}$, then $\worldloweq{(A ; C)}{(B ; D)}{\Hsymbol}$;
\item \label{theorem:eq-loop} If $\FV{A} \cup \FV{B} \subseteq {\Hsymbol}$ and $\worldloweq{A}{B}{\Hsymbol}$, then $\worldloweq{A^*}{B^*}{\Hsymbol}$.
\end{enumerate}
\end{theorem}

Sequential composition (Property~\ref{theorem:eq-decomposition}) is particularly useful. The condition $\FV{C} \cup \FV{D} \subseteq \Hsymbol$ ensures that $\Hsymbol$ includes all variables that might affect the evaluation of programs $C$ and $D$.  We use this property when considering \Hequivalence of $\inlineCode{ctrl; plant}$ and $\attacked{\inlineCode{ctrl; plant}}{S_A}=\attacked{\inlineCode{ctrl}}{S_A}\inlineCode{; plant}$. In particular, if $\Hsymbol$ includes the actuators by which the controller interacts with the physical environment, then $\worldloweq{\inlineCode{ctrl}}{\attacked{\inlineCode{ctrl}}{S_A}}{\Hsymbol}$ ensures that the physical dynamics (i.e., program \inlineCode{plant}) can evolve identically in both the attacked and unattacked systems. 



\newcommand{\programA}{\inlineCode{voting}}
\newcommand{\programB}{\attacked{\programA}{\{v_{s_1}\}}}
\newcommand{\programC}{\alpha}

Consider the previously presented model of an autonomous vehicle with three velocity sensors shown in Figure~\ref{fig:eg-vehicle-robustness-safety},
and let $\programwoattack$ be its hybrid program \inlineCode{(ctrl; plant)$^*$}  
and $\programC$ be program $\programwoattack$ with $\programA$ excluded, i.e., $\programwoattack=(\programA ; \programC)^*$ and $\attacked{\programwoattack}{\{v_{s_1}\}} = (\programB; \programC)^*$. Here, 
$\FV{\programA}$ = $\FV{\programB}$ = $\{ v_p \}$, $\FV{\programC}$ = $\{ v_s, v_p, d_p, A, B, \epsilon \}$, and $\FV{\programA; \programC} = \FV{\programB; \programC}$ = $\{ v_p, d_p,  A, B, \epsilon \}$. 

By definition of $\eqprogram{}{}{\Hsymbol}$, we know
$\eqprogram{\programA}{\programB}{\{v_s, v_p \}}$. By Property~\ref{theorem:eq-unmodified}, we get
\[ \eqprogram{\programA}{\programB}{\FV{\programC}}
\]
Then by Property~\ref{theorem:eq-subset}, 
\[
  \eqprogram{(\programA; \programC)}{(\programB; \programC)}{\FV{\programA; \programC}}
\]
Since $\eqprogram{\programC}{\programC}{\FV{\programA; \programC}}$ (Property~\ref{theorem:eq-self}),  by Property~\ref{theorem:eq-decomposition} we know, 
\[
  \eqprogram{\programA; \programC}{\programB; \programC}{\FV{\programA; \programC}}
\]
By Property~\ref{theorem:eq-loop}, we get
\[
  \eqprogram{(\programA; \programC)^*}{(\programB; \programC)^*}{\FV{\programA; \programC}}
\]
The free variables of $\phi_{pre} \land \phi_{post}$ (shown in Figure~\ref{fig:eg-vehicle-sensing}) are $\{ v_p, d_p, A, B, \epsilon \}$, the same as $\FV{\programA; \programC}$. Thus, by Theorem~\ref{lemma:robust-safety}, we have $\robustsafety{\programwoattack}{\phi_{pre}}{\phi_{post}}{\{v_{s_1}\}}$.

\subsection{Self-Composition Approach} \label{sec:self-comp-prove-ni}

The second approach toward proving $\Hequivalence$ is inspired by self-composition \cite{barthe2004secure,terauchi2005secure}, a proof technique often used for proving noninterference \cite{Sabelfeld2003, goguen1982security}. Noninterference is a well-known strong information security property that, intuitively, guarantees that confidential inputs do not influence observable outputs, or dually guarantees that low-integrity inputs of a system do not affect high-integrity outputs.
%
%
%
Noninterference is a relational property: it compares two executions of a program with different low-integrity inputs. 


To develop an intuition for how the self-composition technique is used to prove noninterference,
consider the problem of checking whether low-integrity inputs of a deterministic program affect high-integrity outputs. 
Construct two copies of the program, renaming the program variables so that the variables in the two copies are disjoint. Set the high-integrity inputs in both copies to identical values but allow the low-integrity inputs to take different values. Now, sequentially compose these two programs together. If the composed program can terminate in a state where the corresponding high-integrity outputs differ, then the original program does not satisfy noninterference; conversely, if in all executions of the composed program, the high-integrity outputs are the same, then the original program satisfies noninterference. Intuitively, the composition of the two copies allows a single program to represent two executions of the original program, reducing checking a relational property of the original problem to checking a safety property of the composed program.

Using the same insights, 
we develop a self-composition technique for hybrid programs, allowing us to use existing verification tools such as  KeYmaera X (which can reason about safety properties of hybrid programs) to reason about \Hequivalence of two hybrid programs.


It is non-trivial to adapt the self-composition approach to hybrid programs due to the nondeterminism in hybrid programs. In particular, to show that two executions of the same hybrid program are in an appropriate relation, it may be necessary to force (some of) the nondeterminism in the two executions to resolve in the same way. For example, a nondeterministic choice in a hybrid program may represent a decision by a driver to brake or accelerate; the driver's decision is assumed to be a high-integrity input, and so the resolution of the nondeterministic choice should be the same in both executions. The self-composition must somehow couple the nondeterministic choices to ensure this. Nondeterministic assignment must be similarly handled, i.e., resolution of high-integrity nondeterminism must be coupled in the two executions.

%
%

An additional source of nondeterminism in hybrid programs is the duration of physical evolution. The program construct for physical dynamics, $x'=\theta \& \evolconstraint$, specifies that the variable(s) evolve according to the differential equation system $x'=\theta$ for an \emph{arbitrary duration} within the region described by formula $\evolconstraint$. The duration is chosen nondeterministically.

\medskip

Our self-composition technique takes as input a program $P$ and set of
sensor variables $S_A$ and creates a program that represents an
execution of each of $P$ and $\attacked{P}{S_A}$. We ensure that the composed program
(1) resolves high-integrity nondeterministic
choices and assignments the same in both executions; and (2)
has the same duration for corresponding physical evolutions.

To ensure that the two executions are appropriately related, we
produce a formula that encodes that the two executions have the same values for high-integrity variables; we assume this formula holds at the beginning of the executions, and require the formula to hold at the end of every control iteration. If we can prove that this is the case, then we have proved that
if the two executions (1) have the same values for high-integrity inputs at the beginning of their executions, (2) follow the same decisions on high-integrity nondeterminism during their executions, and (3) evolve for the same duration, then the two executions have the same values for high-integrity variables at the end of every control iteration.

Our self-composition approach has some limitations on the hybrid programs to which it applies. First, it is applicable only for hybrid programs of the form \inlineCode{(ctrl; plant)$^*$}. Second, it is applicable only for hybrid programs that have \emph{total} semantics for low-integrity inputs. Intuitively, it means if a program has a \emph{valid} execution for an input state $\traceState$ (i.e., exists a state $\traceStateprime$ such that $\HPtransitionPair{\traceState}{\traceStateprime} \in \HPtransition{\hybridprogram}$), then the program has a valid execution for every input state that differs with $\traceState$ only on low-integrity inputs.
The reason for this requirement is that self-composition uses a single program to represent two executions; this composed program has a valid execution only if both executions are valid. Since the two executions differ only on low-integrity inputs, our technique works only if semantics of the unattacked program is total on low-integrity inputs. A straightforward syntactic checker can be developed to check whether a hybrid program meets this requirement. More discussion about the limitation and the syntactic checker can be found in Appendix~\ref{appendix:limitation}.

\smallskip

The rest of this section describes in detail our self-composition approach: how to construct a single program that represents an execution of $P$ and $\attacked{P}{S_A}$, and then prove it correct. At a high level, our approach works by
(1) converting program $P$ to a canonical form $P_\mathit{canon}$ that makes high-integrity nondeterministic choices and assignments explicit; and then
(2) composing $P_\mathit{canon}$ and $\attacked{P_\mathit{canon}}{S_A}$ to ensure that the values of high-integrity nondeterministic choices and assignments, and evolution durations are the \emph{same} for both executions.



\medskip
\paragraph{Canonical Form for Hybrid Programs}



Given a hybrid program of the form $\HPpair{ctrl}{plant}$, we rewrite it to a canonical form $\HPpair{choices; ctrl'}{plant}$ such that
(1) each high-integrity nondeterministic choice $\alpha \cup \beta$ in $ctrl$ is turned into a construct $\mathit{if}~c~\mathit{then}~\alpha~\mathit{else}~\beta$\footnote{Construct $\mathit{if}~\phi~\mathit{then}~\alpha~\mathit{else}~\beta$ is syntactic sugar for \inlineCode{($?\phi$;$\alpha$) $\cup$ ($?\neg\phi$;$\beta$)}.} in $ctrl'$, and (2)  each high-integrity nondeterministic assignment $x:=*$ in $ctrl$ is turned into $x := c$ in $ctrl'$,
where $c$ is a fresh variable, and $choices$ contains a nondeterministic assignment $c:=*$. The program fragment $choices$ consists solely of a sequence of these nondeterministic assignments to these \emph{choice variables}. Note that $\HPpair{ctrl}{plant}$ is semantically equivalent to $\HPpair{choices; ctrl'}{plant}$.  

The goal of the canonical form is to make it easier to share the same nondeterministic choices and assignments between the two executions: when we compose the two programs, they will essentially share the same $choices$ program.




For example, Figure~\ref{fig:eg-canonical-form} shows the previously presented model of an autonomous vehicle with interior temperature control shown in Figure~\ref{fig:eg-vehicle-robustness-guard} whose hybrid program is rewritten to the canonical form (elided contents in Figure~\ref{fig:eg-canonical-form} are the same as Figure~\ref{fig:eg-vehicle-robustness-guard}). The program has a nondeterministic choice variable $c$ that represents a decision to brake or accelerate. This choice is considered high-integrity. 


\begin{figure}
\begin{lstlisting}[basicstyle={\linespread{\linespacing}\ttfamily\scriptsize}]
 ...
 HP choices `\codealign{\codepos}{\HPdef \HPassignS{$c$}{$*$} \label{line:acc-choices-cp1}}`
 HP ctrl `\codealign{\codepos}{\HPdef choices; ctrl$_{t}$; \HPassignment{$v_{s}$}{$v_{p}$} \HPassignment{$d_{s}$}{$d_{p}$}}`
 `\codealign{\codepos}{~~~~($\mathit{if}$ ($c$) $\mathit{then}$ accel $\mathit{else}$ brake); $t:=0$ \label{line:canonical-form-choice}}`
 ...
ProgramVariables.
 B $c$. `\javacommentalign{\cmtpos}{choice variable}`
 ...
\end{lstlisting}
\caption{\systemmodel of an autonomous vehicle with interior temperature control shown in Figure~\ref{fig:eg-vehicle-robustness-guard} whose hybrid program is rewritten to canonical form with a choice variable $c$} 
\label{fig:eg-canonical-form}
\end{figure}



\medskip

\paragraph{Hybrid Program with Renaming}

Note that program $\programwoattack$ and $\attacked{\programwoattack}{S_A}$ have the same set of variables. To compare executions of $\programwoattack$ and $\attacked{\programwoattack}{S_A}$ in a composition, we need to rename \emph{bound variables} in one of the two programs. Renaming is needed only for bound variables, since their values may differ during executions. Other variables are read-only and their values will be the same for executions of program $\programwoattack$ and $\attacked{\programwoattack}{S_A}$. Thus, these variables can be shared by both programs, and renaming is not needed.

To help us with renaming, we define \emph{renaming functions} that map all and only the bound variables of a program to fresh variables.

\begin{definitions}[Renaming function] \label{def:var-rename}  For hybrid program $\hybridprogram$,
  function $\xi : \variableSet{\hybridprogram} \rightarrow V$ (where $V$ is a set of variables) is a \emph{renaming function for $\hybridprogram$} if:
  \begin{enumerate}[itemsep=-1ex]
  \item $\xi$ is a bijection;
  \item For all $x \in \BV{\hybridprogram}$, $\xi(x) \not\in \variableSet{\hybridprogram}$;
  \item For all $x \in \variableSet{\hybridprogram}\setminus \BV{\hybridprogram}$, $\xi(x) = x$.
  \end{enumerate}
\end{definitions}

We write $\renaming{\hybridprogram}{\xi}$ for the program identical to $\hybridprogram$ but whose variables have been renamed according to function $\xi$. We also apply renaming functions to states and formulas, with the obvious meaning.


\medskip
\paragraph{Interleaved Composition}

We develop an \emph{interleaved composition} that composes two programs so their executions have the same values for high-integrity nondeterministic choices and assignments, and last the same evolution duration.

\begin{definitions}[Interleaved composition] \label{def:composition} Given a hybrid program $\programwoattack = (choices;\, ctrl; \, (x'=\theta \,\& \,\evolconstraint))^*$ in canonical form,
a renaming function $\xi$ for $\programwoattack$, a set of variables $S_A$ $\subseteq$ $\BV{\programwoattack}$, the \emph{interleaved composition of $\programwoattack$ under $S_A$ attack with renaming function $\xi$}, denoted $\compsym{\programwoattack}{S_A}{\xi}$, is the following program:
\begin{align*}
 (choices; \,  ctrl; \, & \subs{choices}{\xi}; \, \renaming{\attacked{ctrl}{S_A}}{\xi};  \\
& (x'=\theta, \renaming{x'=\theta}{\xi} \,\& \,\evolconstraint \land \renaming{\evolconstraint}{\xi}))^*
\end{align*}
Where function $\subs{choices}{\xi}$ replaces $c_i := *$ in program $choices$ with $\xi(c_i) := c_i$ for all variables $c_i$ in $\BV{choices}$.
\end{definitions}

The composition has the following properties: (1) control components from two programs are executed sequentially (i.e., $choices; \,  ctrl; \, \subs{choices}{\xi}; \, \renaming{\attacked{ctrl}{S_A}}{\xi}$); (2) plants are executed in parallel (i.e., $x'=\theta, \renaming{x'=\theta}{\xi}$)~\cite{muller2016component}; (3) the evolution constraint is a conjunction of the two evolution constraints (i.e., $\evolconstraint \land \renaming{\evolconstraint}{\xi}$), and (4) nondeterministic choices in $choices$ used by $ctrl$ and their counterparts used by $\renaming{\attacked{ctrl}{S_A}}{\xi}$ have the same values. 

For example, let $\programwoattack$ be the previously presented hybrid program (in canonical form) of an autonomous vehicle with interior temperature control shown in Figure~\ref{fig:eg-canonical-form}.
Figure~\ref{fig:eg-self-composition-more} shows $\compsym{\programwoattack}{\{ temp_s \}}{\xi}$, where function $\xi$ renames bound variables in $\attacked{\programwoattack}{S_A}$ with subscript 1. 
Program \inlineCode{ctrl$^\prime$} and \inlineCode{Plant$^\prime$} compose two programs as described in Definition~\ref{def:composition} (lines~\ref{line:acc-composed-ctrl}--\ref{line:acc-composed-plant2}).
Line~\ref{line:acc-choices-cp2} shows the effect of function $\subs{choices}{\xi}$: substituting \inlineCode{$c = *$} with \inlineCode{$c_1 = c$} in $choices$. The choice represents a decision to accelerate or brake, which is high-integrity. The resolution of this choice should be the same in both executions.


\renewcommand{\linespacing}{1.15}
\renewcommand{\codepos}{6.6cm}

\renewcommand\cmtpos{5.7cm}

\begin{figure}
\begin{lstlisting}[basicstyle={\linespread{\linespacing}\ttfamily\scriptsize}]    
Definitions.
 R $\epsilon$. `\javacommentalign{\cmtpos}{time limit of control}`
 R $A$. `\javacommentalign{\cmtpos}{acceleration rate}`
 R $B$. `\javacommentalign{\cmtpos}{braking rate}`
 R $T$. `\javacommentalign{\cmtpos}{target temperature}`
 B $eq_{\eqSet}$ `\codealign{\codepos}{\HPdef $v_p = v_{p_1} \land d_p = d_{p_1}$ \label{line:acc-eq-formula}}`
 B $\psi$ `\codealign{\codepos}{\HPdef $2Bd_{s} > v_{s}^2 + (A+B)(A\epsilon^2 + 2v_{s}\epsilon)$}`
 HP choices `\codealign{\codepos}{\HPdef \HPassignS{$c$}{$*$}}`
 HP ctrl$_{t}$ `\codealign{\codepos}{\HPdef \HPassignment{$temp_{s}$}{$temp_{p}$}}`
 `\codealign{\codepos}{~~(\phantom{$\cup$~}(\HPguard{$temp_s>T$} \HPassignS{$thermo$}{-1})}`
 `\codealign{\codepos}{~~~$\cup$ (\HPguard{$temp_s<T$} \HPassignS{$thermo$}{1})}`
 `\codealign{\codepos}{~~~$\cup$ ($?temp_s=T$)~)} `
 HP accel `\codealign{\codepos}{\HPdef $?\psi; a:= A$}`
 HP brake `\codealign{\codepos}{\HPdef $a := -B$}`
 HP ctrl `\codealign{\codepos}{\HPdef ctrl$_{t}$; \HPassignment{$v_{s}$}{$v_{p}$} \HPassignment{$d_{s}$}{$d_{p}$}}`
 `\codealign{\codepos}{~~~(if ($c$) then accel else brake); $t:=0$}`
 B $\psi_1$ `\codealign{\codepos}{\HPdef $2Bd_{s_1} > v_{s_1}^2 + (A+B)(A\epsilon^2 + 2v_{s_1}\epsilon)$}`
 HP choices$_1$ `\codealign{\codepos}{\HPdef \HPassignS{$c_1$}{$c$} \label{line:acc-choices-cp2}}`
 HP ctrl$_{t_1}$ `\codealign{\codepos}{\HPdef \HPassignment{$temp_{s_1}$}{*}}`
 `\codealign{\codepos}{~~(\phantom{$\cup$~}(\HPguard{$temp_{s_1}>T$} \HPassignS{$thermo_1$}{-1})}`
 `\codealign{\codepos}{~~~$\cup$ (\HPguard{$temp_{s_1}<T$} \HPassignS{$thermo_1$}{1})}`
 `\codealign{\codepos}{~~~$\cup$ ($?temp_{s_1}=T$)~)} `
 HP accel$_1$ `\codealign{\codepos}{\HPdef $?\psi_1; a_1:= A$}`
 HP brake$_1$ `\codealign{\codepos}{\HPdef $a_1 := -B$}`
 HP ctrl$_1$ `\codealign{\codepos}{\HPdef ctrl$_{t_1}$; \HPassignment{$v_{s_1}$}{$v_{p_1}$} \HPassignment{$d_{s_1}$}{$d_{p_1}$}}`
 `\codealign{\codepos}{~~~(if ($c_1$) then accel$_1$ else brake$_1$); $t_1:=0$}`
 HP ctrl$^\prime$ `\codealign{\codepos}{\HPdef choices; ctrl; choices$_1$; ctrl$_1$ \label{line:acc-composed-ctrl}}`
 HP plant$^\prime$ `\codealign{\codepos}{\HPdef $d_p^\prime = -v_p, v_p^\prime = a, temp_p^\prime = thermo, t^\prime = 1$ \label{line:acc-composed-plant1}}`
 `\codealign{\codepos}{~~~$d_{p_1}^\prime = -v_{p_1}, v_{p_1}^\prime = a_1, temp_{p_1}^\prime = thermo_1, t_1^\prime = 1 \;$}`
 `\codealign{\codepos}{\phantom{~~~~~}$\& \; (v_p \geq 0 \land v_{p_1} \geq 0 \land t \leq \epsilon \land t_1 \leq \epsilon)$ \label{line:acc-composed-plant2}}` 
ProgramVariables.
 B $c, c_1$. `\javacommentalign{\cmtpos}{choice variables}`
 R $t, t_1$. `\javacommentalign{\cmtpos}{clock variables}`
 R $d_p, d_{p_1}$. `\javacommentalign{\cmtpos}{distance to obstacle (physical)}`
 R $d_{s}, d_{s_1}$. `\javacommentalign{\cmtpos}{distance to obstacle (sensed)}`
 R $v_p, v_{p_1}$. `\javacommentalign{\cmtpos}{vehicle velocity (physical)}`
 R $v_{s}, v_{s_1}$. `\javacommentalign{\cmtpos}{vehicle velocity (sensed)}`
 R $a, a_1$. `\javacommentalign{\cmtpos}{acceleration of the vehicle}`
 R $temp_{s}, temp_{s_1}$. `\javacommentalign{\cmtpos}{interior temperature (sensed)}`
 R $temp_p, temp_{p_1}$. `\javacommentalign{\cmtpos}{interior temperature (physical)}`
 R $thermo, thermo_1$. `\javacommentalign{\cmtpos}{rates of change for temperature }`
Problem.
 $eq_{\eqSet}$ $\rightarrow$ [(ctrl$^\prime$; plant$^\prime$)$^*$]$eq_{\eqSet}$ `\label{line:acc-composed-problem}`
\end{lstlisting}
\caption{Interleaved composition of the hybrid program (in canonical form) modeling an autonomous vehicle with interior temperature control shown in Figure~\ref{fig:eg-canonical-form}} \label{fig:eg-self-composition-more}
\end{figure}

\medskip

\paragraph{Proving $\Hequivalence$ with an Interleaved Composition}

%
Given an interleaved composition $\compsym{\programwoattack}{S_A}{\xi}$, to prove that two programs are $\Hequivalent$ on a set $\Hsymbol$, we need to first identify a set $\eqSet$ of high-integrity variables on which the evaluation of variables in $\Hsymbol$ depend. Then we construct a formula to express that the two program executions have the same values for variables in set $\eqSet$, and finally prove that, for any execution of the composition, if the formula holds initially, it would hold at the end of every control loop iteration of the execution. 



\begin{definitions}[Equivalence formula] \label{def:low-eq-expressions}
  For a set ${\eqSet}$ of variables, a renaming function $\xi$ such that $\eqSet \subseteq dom(\xi)$, the equivalence formula of ${\eqSet}$ and $\xi$, denoted $eq_{\eqSet}^{\xi}$, is defined as:
\[
  eq_{\eqSet}^{\xi} \equiv \bigwedge_{x \in \eqSet} (x = \renaming{x}{\xi})
\]
\end{definitions}
Then the desired property is, for any execution of the composition, if the equivalence formula holds at the beginning of an execution, it holds at the end of every control loop iteration of the execution. That means, we want to prove the following: 
\[
  eq_{\eqSet}^{\xi} \rightarrow [\compsym{\programwoattack}{S_A}{\xi}]eq_{\eqSet}^{\xi}
\]

For example, $eq_{\eqSet}$ in Figure~\ref{fig:eg-self-composition-more} (line~\ref{line:acc-eq-formula}) encodes that the two executions have the same position ($d_p=d_{p_1}$) and velocity ($v_p=v_{p_1}$). The desired property is shown at line~\ref{line:acc-composed-problem}.

We have proven this property using Keymaera X. 
Intuitively, proving this property means that for any execution of the autonomous vehicle model, whether or not its temperature sensor is compromised, 
if the vehicle starts with the same position and velocity, makes the same control decisions for acceleration and brake, and runs for the same duration, it would end with the same position and velocity.

\medskip
\paragraph{Soundness}

The soundness theorem links the self-composition approach with proving $\Hequivalence$.
Proof of this theorem is based on trace semantics of hybrid programs~\cite{platzer2007temporal, jeannin2014dtl} and can be found in Appendix~\ref{appendix:sound-proof}.

\begin{theorem}[Soundness of the self-composition approach] \label{theorem:soundness-composition}
For hybrid program $\programwoattack$ and $\programwoattack_{c}$, a set $S_A$ $\subseteq$ $\BV{\programwoattack}$, a renaming function $\xi$ of $\programwoattack_{c}$, a set of variables $\eqSet \subseteq \BV{\programwoattack}$, and a set $\Hsymbol$ $\subseteq$ $\eqSet$, if $\programwoattack_{c}$ is $\programwoattack$ in canonical form, $S_A \cap \eqSet = \emptyset$, and $eq_{\eqSet}^{\xi} \rightarrow [\compsym{\programwoattack_{c}}{S_A}{\xi}]eq_{\eqSet}^{\xi}$, then 
$\eqprogram{\programwoattack}{\attacked{\programwoattack}{S_A}}{\Hsymbol}$.
\end{theorem}

Note that the condition $S_A ~\cap~ \eqSet = \emptyset$ indicates that the adversary cannot compromise high-integrity variables. 




\medskip
\paragraph{Applicability}
Our self-composition technique applies to a subset of problems of interest rather than general problems. In particular, our technique requires that two executions having the same duration at every control iteration for the plant, and identical values for high-integrity nondeterministic assignments.
Our self-composition technique \emph{cannot} be applied to compare two executions that evolve for different durations or that resolve high-integrity nondeterministic choices and high-integrity nondeterministic assignments differently.
However, these restrictions arise naturally for many systems.
First, requiring the same duration of evolution for both executions corresponds to the control system having the same frequency of operation. That is, the rate of the the system's control can't be influenced by the attacker.
%
Second, high-integrity non-deterministic choices and high-integrity non-deterministic assignments are used to model exactly the nondeterminism that cannot be influenced by the attacker. As such, they should be resolved the same in both executions. For example, when considering how a corrupted temperature sensor can affect a (non-autonomous) vehicle, the driver's decisions (i.e., whether to accelerate or brake) would be modeled with a high-integrity nondeterministic choice, since we are concerned with understanding the relationship between two executions where the driver makes the same decisions but in one execution the sensor is corrupted. If in the two executions the driver is making different choices, the two executions might diverge almost arbitrarily, even if the corrupted sensor has no security impact.
If, on the other hand, we want to use this technique to determine whether an autonomous vehicle's driving subsystem can be influenced by a corrupted temperature sensor, we would need a more precise model of the system that does not use nondeterministic choice between accelerating and braking to model the driving subsystem's decisions. That is, high-integrity nondeterministic choices are by assumption choices that cannot be influenced by the attacker.




\section{Case Studies} \label{sec:casestudies}
To demonstrate the feasibility and efficacy of our approach, we conduct three case studies of non-trivial systems. The first two case studies analyze robustness of safety with the decomposition approach, and the third one proves robustness of high-integrity state with the self-composition approach.

\subsection{Case Study: an Anti-lock Braking System}

System designers may wonder if the system is robustly safe against sensor attacks or if their countermeasures are effective. 
This case study demonstrates analyzing robustness of safety with the decomposition approach in an Anti-lock Braking System (ABS). 
An ABS is a safety braking system used on aircraft and vehicles. It operates by preventing the wheels from locking up during braking, thereby maintaining tractive contact with the road surface. ABS monitors the speed of wheels using the wheel-speed sensors. If the controller sees that one wheel is decelerating at a rate that couldn't possibly correspond to the vehicle's rate of deceleration, it reduces the brake pressure applied to that wheel, which allows it to turn faster. Once the wheel is back up to speed, it applies the brake again \cite{Nielsen_2000}. 

\medskip

\paragraph{Modeling ABS} Figure~\ref{fig:hp-abs} shows a model of an ABS system\cite{tanelli2008robust, bhivate2011modelling}. The model assumes a single wheel
and uses a simple controller that turns on and off maximum braking torque. Intuitively, ABS systems are designed to achieve the maximum friction under certain circumstances (e.g., braking on icy road surface). They achieve this by maintaining an ideal slip ratio (e.g.,  \inlineCode{$\lambda_{ref}$} in Figure~\ref{fig:hp-abs}). Our controller switches the brake on and off based on the calculated slip ratio (\inlineCode{$\lambda_c$}) and reference slip ratio (lines~\ref{line:hp-abs-ctrl1}--\ref{line:hp-abs-ctrl3}). The calculated slip ratio  is computed using sensed wheel speed and vehicle speed (lines \ref{line:hp-abs-cal1}--\ref{line:hp-abs-cal2}). The physical slip ratio (\inlineCode{$\lambda_p$}) depends on physical wheel speed ($v_p$) and vehicle speed ($w_p$), which are affected by braking torque (\inlineCode{$T_b$}) and adhesion coefficient (\inlineCode{$\mu_p$}) that depends on the physical slip ratio (line~\ref{line:hp-abs-cal3}). 

\renewcommand\codepos{6.9cm}
\renewcommand\cmtpos{6.5cm}

\begin{figure}
  \begin{center}
\begin{lstlisting}[basicstyle=\linespread{\linespacing}\ttfamily\scriptsize]
Definitions.
 HP ctrl `\codealign{\codepos}{\HPdef \HPassignment{$\omega_s$}{$\omega_p$} \HPassignment{$v_s$}{$v_p$}}` `\label{line:hp-abs-cal1}`
 `\codealign{\codepos}{~~~\HPassignment{$\lambda_c$}{$\cfrac{v_s - \omega_s * R}{v_s}$} \HPassignment{$\lambda_p$}{$\cfrac{v_p - \omega_p * R}{v_p}$}}`  `\label{line:hp-abs-cal2}`
 `\codealign{\codepos}{~~ \HPassignment{$\mu_p$}{$C_1(1-e^{-C_2\lambda_p})-C_3\lambda_p$}}`  `\label{line:hp-abs-cal3}`
 `\codealign{\codepos}{~~(\phantom{$\cup$~}($\HPguard{\lambda_c < \lambda_{ref}}$ \HPassignment{BRAKE}{0} \HPassignS{$T_b$}{0})}` `\label{line:hp-abs-ctrl1}`
 `\codealign{\codepos}{~~~$\cup$ ($\HPguard{\lambda_c = \lambda_{ref}}$ ?True)}` `\label{line:hp-abs-ctrl2}`
 `\codealign{\codepos}{~~~$\cup$ ($\HPguard{\lambda_c > \lambda_{ref}}$ \HPassignment{BRAKE}{1} \HPassignS{$T_b$}{1200})~); $\HPassignS{t}{0}$}` `\label{line:hp-abs-ctrl3}`
 HP Plant `\codealign{\codepos}{\HPdef $v_p'=\cfrac{- \mu_p F_N }{m}, \omega_p'=\cfrac{\mu_p F_N R- T_b}{J}, t'=1$}` `\label{line:hp-abs-plant1}`
 `\codealign{\codepos}{\phantom{~~~~~}$\&~v_p \geq 0 \land \omega_p \geq 0 \land t \leq \epsilon$}` `\label{line:hp-abs-plant2}`
 B $\phi_\mathit{pre}$ `\codealign{\codepos}{\HPdef ($v_p$ = 100 $\land$ $\omega_p$ $\geq$ 0)}` `\label{line:hp-abs-init}`
 B $\phi_\mathit{post}$ `\codealign{\codepos}{\HPdef ($v_p$ > 25 $\rightarrow$ $\omega_p$ $\geq$ 1)}` `\label{line:hp-abs-post}`
 R $\epsilon$. `\javacommentalign{\cmtpos}{control interval}`
 R $C_1, C_2, C_3.$ `\javacommentalign{\cmtpos}{constant for computing $\mu$}`
 R $J, R.$ `\javacommentalign{\cmtpos}{wheel inertia and wheel radius}`
 R $F_N, m.$ `\javacommentalign{\cmtpos}{normal force and vehicle mass}`
 R $\lambda_{ref}.$ `\javacommentalign{\cmtpos}{reference value of wheel slip ratio}`
ProgramVariables.
 R BRAKE. `\javacommentalign{\cmtpos}{brake status}`
 R $T_b$. `\javacommentalign{\cmtpos}{braking torque}`
 R $\omega_p, \omega_s$. `\javacommentalign{\cmtpos}{wheel speed (physical and sensed )}`
 R $v_p, v_s$. `\javacommentalign{\cmtpos}{vehicle speed (physical and sensed)}`
 R $\lambda_p, \lambda_c$. `\javacommentalign{\cmtpos}{wheel slip (physical and calculated)}`
 R $\mu_p$. `\javacommentalign{\cmtpos}{adhesion coefficient}`
 R $t$. `\javacommentalign{\cmtpos}{clock variable}`
Problem.
  $\phi_\mathit{pre} \rightarrow [(\text{ctrl}; \text{plant})^*]\phi_\mathit{post}$
\end{lstlisting}
\end{center}
\caption{\systemmodel of an ABS system}
\label{fig:hp-abs}
\end{figure}

The initial condition of the ABS system ($\phi_{pre}$) is that the vehicle is moving at a high speed and its wheel speed is not negative (line~\ref{line:hp-abs-init}). 
The safety condition 
($\phi_{post}$) is that the vehicle's wheel should not lock if the current vehicle speed is large (line \ref{line:hp-abs-post}) \cite{solyom2004synthesis}. 

\smallskip

\paragraph{Modeling Non-invasive Attack on ABS}

Previous research has demonstrated attacks on ABS through physical channels \cite{shoukry2013non}. By placing a thin electromagnetic actuator near the ABS wheel-speed sensors, an attacker can inject magnetic fields to both cancel the true measured signal and inject a malicious signal, thus spoofing the measured wheel speeds.
Such an attack is a \attackname, where $S_A=\{\omega_s\}$, on the wheel-speed sensor. Let $\programwoattack$ be the hybrid program modeling an ABS system shown in Figure~\ref{fig:hp-abs}. Then $\attacked{\programwoattack}{\{ \omega_s \}}$ is program $\programwoattack$ with line~\ref{line:hp-abs-cal1} changed into the following:
\[
\texttt{ctrl}  \equiv \highlight{\HPassignment{\omega_s}{*}} ~ \HPassignS{v_s}{v_p}
\]
Program $\programwoattack$ is not robustly safe when the sensor $\omega_s$ is compromised: assuming $\safety{\programwoattack}{\phi_{pre}}{\phi_{post}}$ holds, $\safety{\attacked{\programwoattack}{\{\omega_s\}}}{\phi_{pre}}{\phi_{post}}$ doesn't necessarily hold, since $\omega_s$ can be an arbitrary value.


\smallskip

\paragraph{Designing Robustly Safe ABS System}
System designers, in attempts to make ABS system modeled in Figure~\ref{fig:hp-abs} safer, would be confident in their design if the system with countermeasures can be proven to be robustly safe against the attack. 

Assume that designers deploy three wheel-speed sensors and a majority voting scheme in the ABS system modeled in Figure~\ref{fig:hp-abs}. The countermeasure can be modeled by changing line~\ref{line:hp-abs-cal1} in Figure~\ref{fig:hp-abs} into the \inlineCode{ctrl $\, \equiv \,$ voting; $~\HPassignment{v_s}{v_p}$}, where \inlineCode{voting} is the following: 
\begin{align*}
  \texttt{voting}
  \equiv \; & \HPassignment{\omega_{s_1}}{\omega_p} \; \HPassignment{\omega_{s_2}}{\omega_p} \; \HPassignment{\omega_{s_3}}{\omega_p}  \\
         & \texttt{if~} (\omega_{s_1} = \omega_{s_2} \lor \omega_{s_1} = \omega_{s_3}) \\
         & \texttt{then~} \HPassignS{\omega_{s}}{\omega_{s_1}} \texttt{~else~} \HPassignS{\omega_{s}}{\omega_{s_2}}                         
\end{align*}
Using the decomposition approach, we can prove that such an ABS system is robustly safe if only one wheel-speed sensor is compromised. The proof can be found in Appendix~\ref{appendix:proofs}.


\subsection{Case study: Boeing 737-MAX}
Robustness of safety is a relational property: if the original system is safe then the attacked system will be safe too. 
Importantly, this separates reasoning about the implications of sensor attacks from reasoning directly about safety properties.
Proving a system's safety is often labor-intensive and may even be epistemically problematic. 
For example, many systems must be verified and validated empirically because their correctness properties are not possible to state in a formal language.
However, when it is not easy or even impossible to formally verify safety, it is often still possible to prove that compromised sensors do not affect the safety property. 

To demonstrate this advantage of relational reasoning, we present a case study inspired by the Boeing 737-MAX Maneuvering Characteristics Augmentation System (MCAS)~\cite{mcaswiki}. The 737-MAX's dynamics are extremely complicated, and proving properties about similar stabilization systems is an open challenge in hybrid systems verification \cite{heidlauf2018verification}. 
Nonetheless, we are able to analyze robustness of safety against faults or attacks on the angle of attack (AOA) sensor used by the 737-MAX MCAS, even without an analysis of the system's overall safety property or the MAX's flight dynamics. 

\smallskip

\paragraph{Modeling MCAS}

The MCAS
caused at least two deadly crashes in 2019\cite{boeing2}.
MCAS was added to compensate for instability induced by the 737-MAX's new engines.
Adding new engines to an existing airframe resulting in an aircraft whose nose tended to pitch upward, risking stalls.
The MCAS adjusts the plane's horizontal stabilizer in order to push the nose down when the aircraft is operating in manual flight at an elevated angle of attack (AOA). 
In many 737-MAX planes, the MCAS is activated by inputs from only one of the airplane's two angle of attack sensors. In both 2019 crashes, the MCAS was triggered repeatedly due to a failed AOA sensor. These false readings caused the MCAS software to repeatedly push the plane's nose down, ultimately interacting with manual inputs in a way that caused violently parabolic flight paths terminating in lost altitude and an eventual crash.

Figure~\ref{fig:hp-boeing737max} shows a simplified model of the original MCAS. 
The controller, plane's flight dynamics, and manual control inputs are all left abstract:
the model focuses only on how values read by the left and right AOA sensors are used in MCAS. On each control iteration, one of the two AOA sensors is randomly chosen (\inlineCode{ctrl$_{aoa}$}) and the MCAS is activated using the value of the chosen sensor (line~\ref{line:737maxmcas}). 
In this model, we intentionally omit details about the flight controller, MCAS system, and flight dynamics. Even with a high-fidelity model\cite{marcos2001linear}, proving correctness for the 737-MAX MCAS requires advances in state-of-the-art reachability analysis for hybrid time systems; fortunately, relational reasoning allows us to nonetheless analyze robustness of the system against faults or attacks on the AOA sensors.

\renewcommand\codepos{6.5cm}
\renewcommand\cmtpos{6.5cm}

\begin{figure}
\begin{lstlisting}[basicstyle={\linespread{\linespacing}\ttfamily\scriptsize}]
Definitions.
 B $\phi_\mathit{pre}$    `\javacommentalign{\cmtpos}{preconditions (abstract)}`
 B $\phi_\mathit{post}$    `\javacommentalign{\cmtpos}{functional safety property (abstract)}`
 HP plant. `\javacommentalign{\cmtpos}{plane's dynamics (abstract)}`
 HP MCAS. `\javacommentalign{\cmtpos}{MCAS actuation (abstract)}`
 HP ctrl$_{aoa}$ `\codealign{\codepos}{$\equiv$ (\quad($s_L := aoa_p$; $s_R := *$)  \label{line:chooseAOA1}}`
 `\codealign{\codepos}{~~~$\cup$~($s_L := *$; $s_R := aoa_p$)~);}`
 `\codealign{\codepos}{~~~($aoa_s := s_L$ $\cup$ $aoa_s := s_R$) \label{line:chooseAOA2}}`
 HP ctrl `\codealign{\codepos}{$\equiv$ ctrl$_{aoa}$; MCAS($aoa_s$) \label{line:737maxmcas}}`
ProgramVariables.
 R $aoa_p$.    `\javacommentalign{\cmtpos}{physical AOA}`
 R $s_L, s_R$.       `\javacommentalign{\cmtpos}{left and right AOA sensor}`
 R $aoa_s$.      `\javacommentalign{\cmtpos}{AOA used by MCAS}`
Problem.
 $\phi_\mathit{pre} \rightarrow$ [(ctrl;plant)$^*$]$\phi_\mathit{post}$
\end{lstlisting}
\caption{A Simple Model of Boeing737 Max flawed MCAS.}
 \label{fig:hp-boeing737max}
\end{figure}

%
%

\medskip

\paragraph{Reasoning for Robustness of Safety}


Program \inlineCode{ctrl$_{aoa}$} is not robustly safe if either of the AOA sensors is compromised, since $aoa_s$ can have false readings. Therefore, the system is not robustly safe for attacks on AOA sensors. 
Boeing's proposed fix to MCAS includes a requirement that the controller should compare inputs from both AOA sensors~\cite{boeingfix}, which can be modeled by adding the following at the end of \inlineCode{ctrl$_{aoa}$}:
\[
  (?s_L = s_R) ~\cup~ (? \neg (s_L = s_R); \, aoa_s := 0) 
\]
We can prove that the system with this fix is robustly safe. Let \inlineCode{ctrl$_{aoa}^\prime$} be \inlineCode{ctrl$_{aoa}$} with this simple fix. Then we know
$\eqprogram{\inlineCode{ctrl}_{aoa}^\prime}{\attacked{\inlineCode{ctrl}_{aoa}^\prime}{\{ s_L \}}}{\{aoa_s\}}$ and 
$\eqprogram{\inlineCode{ctrl}_{aoa}^\prime}{\attacked{\inlineCode{ctrl}_{aoa}^\prime}{\{ s_R \}}}{\{aoa_s\}}$.
Program \HPgeneralform with \inlineCode{ctrl$_{aoa}^\prime$}  is robustly safe by the decomposition approach. The proof is included in Appendix~\ref{appendix:proofs}.






\subsection{Case Study: An Autonomous Vehicle with an Internal Bus}
Figure~\ref{fig:eg-self-composition-more} shows the self-composition approach with a model of an autonomous vehicle with interior temperature control. However, the model doesn't account for any internal communication mechanisms.
In modern vehicles, Electronic Control Units (ECUs) oversee a broad range of functionality, including the drivetrain, lighting, and entertainment. They often communicate through an \emph{internal bus} \cite{checkoway2011comprehensive}. 

In this case study, we explore how to use the self-composition approach to analyze robustness of high-integrity state in a model of an autonomous vehicle with an internal bus that communicates both low-integrity messages (sensed temperature) and a high-integrity messages (sensed velocity). We are interested in whether the high-integrity state (i.e., velocity) is robust when the temperature sensor is compromised. 

\medskip

\paragraph{Modeling a Vehicle with an Internal Bus}

\renewcommand\codepos{6.6cm}
\renewcommand\cmtpos{7cm}

\begin{figure} \begin{lstlisting}[basicstyle={\linespread{\linespacing}\ttfamily\scriptsize}]    
Definitions.
 ...
 HP accel `\codealign{\codepos}{\HPdef $?\psi; busV:= A$}`
 HP brake `\codealign{\codepos}{\HPdef $busV := -B$}`
 HP ctrl$_{v}$ `\label{line:bus-ctrl-speed}` `\codealign{\codepos}{\HPdef \HPassignment{$v_{s}$}{$v_{p}$} \HPassignment{$d_{s}$}{$d_{p}$}}`
 `\codealign{\codepos}{~~~if ($c$) then accel else brake}`
 HP ctrl$_{t}$ `\codealign{\codepos}{\HPdef \HPassignment{$temp_{s}$}{$temp_{p}$} \label{line:bus-ctrl-temp-read}}`
 `\codealign{\codepos}{~~(\phantom{$\cup$~}(\HPguard{$temp_{s}>T$} \HPassignS{$busV$}{-1})}`
 `\codealign{\codepos}{~~~$\cup$~(\HPguard{$temp_{s}<T$} \HPassignS{$busV$}{1})}`
 `\codealign{\codepos}{~~~$\cup$~($?temp_{s}=T$)~)}`
 HP ctrl$_{bus}$ `\label{line:bus-temp-ctrl1}` `\codealign{\codepos}{\HPdef (~~(\HPguard{$busV = a$} ctrl$_{t}$;  \HPassignS{$busH$}{$1$})}`
 `\codealign{\codepos}{~~~$\cup$~~\HPassignS{$busH$}{$0$}~) \label{line:bus-temp-ctrl2}}`
 HP ctrl$_{r}$ `\codealign{\codepos}{\HPdef (~~(\HPguard{$busH = 0$} \HPassignS{$a$}{$busV$}) \label{line:bus-read1}}`
 `\codealign{\codepos}{~~~$\cup$~(\HPguard{$busH = 1$} \HPassignS{$thermo$}{$busV$})~) \label{line:bus-read2}}`
 HP ctrl `\codealign{\codepos}{\HPdef choices; ctrl$_{v}$; ctrl$_{bus}$; ctrl$_{r}$;  \HPassignS{$t$}{0}}` 
 HP plant `\label{line:bus-plant}` `\codealign{\codepos}{\HPdef $d_p^\prime = -v_p, v_p^\prime = a, temp^\prime = thermo, t^\prime = 1$}`
 `\codealign{\codepos}{\phantom{~~~~~~}$\& \; (v_p \geq 0 \land t \leq \epsilon)$}`
ProgramVariables.
 R $busV$. `\javacommentalign{\cmtpos}{value on the bus}`
 R $busH$. `\javacommentalign{\cmtpos}{header indicating the type of information}`
 ...
Problem.
 $\phi_\mathit{pre}$ $\rightarrow$ [(ctrl; plant)$^*$]$\phi_\mathit{post}$ 
\end{lstlisting}
\caption{\systemmodel (in canonical form) of an autonomous vehicle with an internal bus} \label{fig:vehicle-bus}
\end{figure}

Figure~\ref{fig:vehicle-bus} shows a model (in canonical form) of such a system (elided contents are the same as in the model previously presented in Figure~\ref{fig:eg-self-composition-more}).
%
We model the bus using two variables: a value variable ($busV$), which indicates the current value that sits on the bus, and a header variable ($busH$), which indicates the type of information that sits on the bus: $busH = 0$ for acceleration, $busH = 1$ for temperature. Exactly one message is communicated via the bus at each control loop iteration.
%
Acceleration messages have higher priority over thermostat messages. Program (\inlineCode{ctrl$_v$}) first sets the bus value to the next acceleration value.
Program \inlineCode{ctrl$_{bus}$} then checks if the value has changed from the existing acceleration value. If not, it activates temperature control (\inlineCode{ctrl$_t$}) to set the bus value to desired thermostat value (line~\ref{line:bus-temp-ctrl1}). Otherwise, $busH$ is sent to $0$ to indicate that a new acceleration value has arrived (line~\ref{line:bus-temp-ctrl2}).
Program \inlineCode{ctrl$_r$} reads a value off the bus and sets corresponding values based on the header (lines~\ref{line:bus-read1}--\ref{line:bus-read2}).

\medskip

\paragraph{Robust High-Integrity State}
We are interested in whether the vehicle's high-integrity state---$v_p$, the velocity of the vehicle---is robust when its low-integrity sensor ($temp_s$) is compromised. Specifically, we wonder whether $\eqprogram{P}{\attacked{\programwoattack}{\{ temp_s \}}}{\{v_p\}}$, where $\programwoattack$ is the model shown in Figure~\ref{fig:vehicle-bus}.
We can prove this using the self-composition approach.
Figure~\ref{fig:bus-vehicle-composed} shows  $\compsym{\programwoattack}{\{temp_s\}}{\xi}$, where $\xi$ renames variables in $\BV{\programwoattack}$ with a subscript 1 (we elide the descriptions of program variables introduced in Figure~\ref{fig:vehicle-bus}). 
By choosing the equivalence formula as  \inlineCode{$v_p = v_{p_1} \land d_p = d_{p_1} \land a = a_1$} (line~\ref{line:bus-composed-eq}), we are able to prove the desired property at line~\ref{line:bus-composed-problem}. 
Proving this property means for this vehicle, its high-integrity variable $v_p$, $d_p$, and $a$ are robust when its temperature sensor is compromised.
%
%
We have proven the model in Figure~\ref{fig:bus-vehicle-composed} using KeYmaera X. 


Note that the decomposition approach and self-composition approach may work well in different settings. The decomposition approach is easy to apply and works well when the sensor attack affects a small portion of the system, as in our first two case studies; by constrast, the self-composition approach can handle cases where the effect of the attack may be complicated---as in our third case study---but requires more effort to use. It is possible to combine the two techniques to prove robustness properties of complicated cases. For example, if we can identify that only a single component of a large system is affected by an attack, the self-composition approach can be used to prove robustness of this component, while the decomposition approach delivers the robustness proof of the whole system.

\renewcommand\codepos{6.5cm}
\renewcommand\cmtpos{5.8cm}

\begin{figure}
  \begin{lstlisting}[basicstyle=\linespread{\linespacing}\ttfamily\scriptsize]    
Definitions.
 ...
 B $eq_{\eqSet}$ `\codealign{\codepos}{\HPdef $v_{p} = v_{p_1} \land d_{p} = d_{p_1} \land a = a_1$ \label{line:bus-composed-eq}}`
 B $\psi$ `\codealign{\codepos}{\HPdef $2Bd_{s} > v_{s}^2 + (A+B)(A\epsilon^2 + 2v_{s}\epsilon)$}`
 HP choices `\codealign{\codepos}{\HPdef \HPassignS{$c$}{$*$}}`
 HP accel `\codealign{\codepos}{\HPdef $?\psi; busV:= A$}`
 HP brake `\codealign{\codepos}{\HPdef $busV := -B$}`
 HP ctrl$_{v}$ `\codealign{\codepos}{\HPdef \HPassignment{$v_{s}$}{$v_{p}$} \HPassignment{$d_{s}$}{$d_{p}$}}`
 `\codealign{\codepos}{~~~if ($c$) then accel else brake}`
 HP ctrl$_{t}$ `\codealign{\codepos}{\HPdef \HPassignment{$temp_{s}$}{$temp_{p}$}}`
 `\codealign{\codepos}{~~(\phantom{$\cup$~}(\HPguard{$temp_{s}>T$} \HPassignS{$busV$}{-1})}`
 `\codealign{\codepos}{~~~$\cup$~(\HPguard{$temp_{s}<T$} \HPassignS{$busV$}{1})}`
 `\codealign{\codepos}{~~~$\cup$~($?temp_{s}=T$)~)} `
 HP ctrl$_{bus}$ `\codealign{\codepos}{\HPdef (~~(\HPguard{$busV = a$} ctrl$_{t}$;  \HPassignS{$busH$}{$1$})}`
 `\codealign{\codepos}{~~~$\cup$~~\HPassignS{$busH$}{$0$}~)}`
 HP ctrl$_{r}$ `\codealign{\codepos}{\HPdef (~~(\HPguard{$busH = 0$} \HPassignS{$a$}{$busV$})}`
 `\codealign{\codepos}{~~~$\cup$~(\HPguard{$busH = 1$} \HPassignS{$thermo$}{$busV$})~)}`
 HP ctrl `\codealign{\codepos}{\HPdef ctrl$_{v}$; ctrl$_{bus}$; ctrl$_{r}$;  \HPassignS{$t$}{0}}` 
 B $\psi_1$ `\codealign{\codepos}{\HPdef $2Bd_{s_1} > v_{s_1}^2 + (A+B)(A\epsilon^2 + 2v_{s_1}\epsilon)$}`
 HP choices$_1$ `\codealign{\codepos}{\HPdef \HPassignS{$c_1$}{$c$} \label{line:bus-composed-choice-same}}`
 HP accel$_1$ `\codealign{\codepos}{\HPdef $?\psi_1; busV_1:= A$}`
 HP brake$_1$ `\codealign{\codepos}{\HPdef $busV_1 := -B$}`
 HP ctrl$_{v_1}$ `\codealign{\codepos}{\HPdef \HPassignment{$v_{s_1}$}{$v_{p_1}$} \HPassignment{$d_{s_1}$}{$d_{p_1}$}}`
 `\codealign{\codepos}{~~~if ($c_1$) then accel$_1$ else brake$_1$}`
 HP ctrl$_{t_1}$ `\codealign{\codepos}{\HPdef \HPassignment{$temp_{s_1}$}{*}}`
 `\codealign{\codepos}{~~(\phantom{$\cup$~}(\HPguard{$temp_{s_1}>T$} \HPassignS{$busV_1$}{-1})}`
 `\codealign{\codepos}{~~~$\cup$~(\HPguard{$temp_{s_1}<T$} \HPassignS{$busV_1$}{1})}`
 `\codealign{\codepos}{~~~$\cup$~($?temp_{s_1}=T$)~)} `
 HP ctrl$_{bus_1}$ `\codealign{\codepos}{\HPdef (~~(\HPguard{$busV_1 = a_1$} ctrl$_{t_1}$;  \HPassignS{$busH_1$}{$1$})}`
 `\codealign{\codepos}{~~~$\cup$~~\HPassignS{$busH_1$}{$0$}~)}`
 HP ctrl$_{r_1}$ `\codealign{\codepos}{\HPdef (~~(\HPguard{$busH_1 = 0$} \HPassignS{$a_1$}{$busV_1$})}`
 `\codealign{\codepos}{~~~$\cup$~(\HPguard{$busH_1 = 1$} \HPassignS{$thermo_1$}{$busV_1$})~)}`
 HP ctrl$_1$ `\codealign{\codepos}{\HPdef ctrl$_{v_1}$; ctrl$_{bus_1}$; ctrl$_{r_1}$;  \HPassignS{$t_1$}{0}}`
 HP ctrl$^\prime$ `\codealign{\codepos}{\HPdef choices; ctrl; choices$_1$; ctrl$_{1}$}` 
 HP plant$^\prime$ `\label{line:bus-composed-plant}` `\codealign{\codepos}{\HPdef $d_p^\prime = -v_p,  v_p^\prime = a, temp_p^\prime = thermo, t^\prime = 1$}`
 `\codealign{\codepos}{~~~$d_{p_1}^\prime = -v_{p_1},  v_{p_1}^\prime = a_1, temp_{p_1}^\prime = thermo_1, t_1^\prime = 1$}`
 `\codealign{\codepos}{\phantom{~~~~~}$\& \; (v_p \geq 0 \land v_{p_1} \geq 0 \land t \leq \epsilon \land t_1 \leq \epsilon)$}`
 ...
Problem.
 $eq_\eqSet$ $\rightarrow$ [(ctrl$^\prime$; plant$^\prime$)$^*$]$eq_\eqSet$ `\label{line:bus-composed-problem}`
\end{lstlisting}
\caption{Interleaved composition of the hybrid program (in canonical form) modeling an autonomous vehicle with an internal bus shown in Figure~\ref{fig:vehicle-bus}} \label{fig:bus-vehicle-composed}
\end{figure}






\section{Related work} \label{sec:relatedwork}

\textbf{Formal analysis of \attackcategory}
Lanotte et al. \cite{lanotte2017formal, lanotte2020formal} propose formal approaches to model and analyze sensor attacks with a process calculus. The threat model allows attacks that manipulate sensor readings or control commands to compromise state. 
Their model of physics is discrete and it focuses on timing aspects of attacks on sensors and actuators. In comparison, we analyze relational properties in systems whose dynamics are modeled with differential equations and we introduce techniques to establish proofs of these properties. 

Bernardeschi et al. \cite{bernardeschi2020formalization} introduce a framework to analyze the effects of attacks on sensors and actuators. Controllers of systems are specified using the formalism PVS~\cite{owre1992pvs}. The physical parts are assumed to be described by other modeling tools.
Their threat model is similar to ours: the effect of an attack is a set of assignments to the variables defined in the controller. Simulation is used to analyze effects of attacks. 
By contrast, we focus on formal analysis for the whole system and propose concrete proof techniques for relational properties.


\textbf{Analyzing relational properties of cyber-physical systems}
Akella et al. \cite{akella2013verification} use
trace-based analysis and apply model checking to verify information-flow properties for discrete models based on process algebra.
Prabhakar et al. \cite{prabhakar2013verifying} introduce a type system that enforces
noninterference for a hybrid system modeled as a programming language.
Nguyen et al. \cite{nguyen2019detecting} propose a static analysis that checks noninterference for hybrid automata.
%
Liu et al. \cite{liu2018secure} introduce an integrated architecture to provide provable security and safety assurance for cyber-physical systems. They focus on integrated co-development: language-based information-flow control using Jif~\cite{myers1999jflow} and a verified hardware platform for information-flow control. Their focus is not on sensor attacks. 

Bohrer et al. \cite{Bohrer2018} verify nondeducibility in hybrid programs, a noninterference-like guarantee. To do this, they introduce a very expressive modal logic that can explicitly express that formulas hold in a given world (i.e., state). By contrast, we use an existing logic (that has good tool support) to express and reason about a specific threat model.


Closely related to our work is that of Kol{\v{c}}{\'a}k et al. \cite{kolvcak2020relational} which introduces a relational extension of \differentiallogic.
A key contribution of their work is a new proof rule to combine two dynamics, allowing existing inference rules of \differentiallogic to be applied in a relational setting.
Similar to their work, our self-composition technique expresses relational properties by leveraging a composition of two programs whose variables are disjoint.
Unlike their work, our self-composition technique aims to prove relational properties that require some of the nondeterministic choices to be resolved in the same way in both executions.
For instance, our example shown in Figure~\ref{fig:eg-self-composition-more} is not directly expressible in their setting.
 We believe that the work by Kol{\v{c}}{\'a}k et al. \cite{kolvcak2020relational} is orthogonal to ours, and the two can be combined to express and prove more complicated relational properties. 


 

\textbf{Security analysis for CPSs}
Much work have focused on the security of cyber-physical systems (CPS), but
primarily from a systems security perspective rather than using formal
methods. Various attacks (and mitigations of these attacks) have been identified, including 
false data injection
\cite{smartcar_vuler}, replay attacks \cite{li2011hijacking}, relay attacks \cite{francillon2011relay},
spying \cite{Checkoway2011}, and
hijacking \cite{langner2011stuxnet}. Our work focuses on formal
methods for CPS security, ruling out entire classes of attacks.

\textbf{Mitigating sensor attacks}
Some work propose attack-resilient state estimation to defend against adversarial sensor attacks in cyber-physical systems~\cite{pajic2014robustness, pajic2016attack}. These methods model systems with bounded sensor noises as an optimization problem to locate potentially malicious sensors. Our work has a different formal model of sensor attacks and focuses on formal guarantees of robustness of systems under sensor attacks.


\section{Conclusion} \label{sec:conclusion}

We have introduced a formal framework for modeling and analyzing \attackcategory on cyber-physical systems.
We formalize two relational properties that relate executions in the original system and a system where some sensors have been compromised. The relational properties express the robustness of safety properties and the robustness of high-integrity state.

Both relational properties can be expressed in terms of an equivalence relation between programs, and we presented two approaches to reason about this equivalence relation, one based on decomposition and the other based on using a single program to represent executions of the original system and the attacked system. We have shown both of these approaches sound, and used them on three case studies of non-trivial cyber-physical systems.



This work focuses on sensors, but our approach can also be used to model and analyze attacks on actuators. 





\let\oldbibliography\thebibliography
\renewcommand{\thebibliography}[1]{%
  \oldbibliography{#1}%
  \setlength{\itemsep}{-5pt}%
}

\bibliographystyle{IEEEtran}    
\bibliography{securityanalysis,bib}

\appendices

\section{Definitions} \label{appendix:definitions}

\setlength{\abovedisplayskip}{1pt}
\setlength{\belowdisplayskip}{1pt}

We present the formal definitions of bound variables, free variables, and variable sets here. These definitions are exactly as given in \cite{Platzer18book,platzer2017complete}, and included for the reader's convenience.
 
\begin{definitions}[Bound variables] The set $\BV{\phi}$ of bound variables of \differentiallogic formula $\phi$ is defined inductively as:
   \begin{align*}
     \BV{\theta_1 \sim \theta_2} & = \emptyset  ~~~~\sim \in \{ <, \leq, =, >, \geq \} \\
     \BV{\neg \phi} & = \BV{\phi} \\
     \BV{\phi \lor \psi}  = \BV{\phi \land \psi}    &  = \BV{\phi} \cup \BV{\psi} \\
    \BV{\phi \rightarrow \psi}   &  = \BV{\phi} \cup \BV{\psi} \\
    \BV{\forall x.~ \phi} = \BV{\exists x.~ \phi}   &  = \{ x \} \cup \BV{\phi} \\
    \BV{\HPbox{\alpha}\phi}   &  = \BV{\alpha} \cup \BV{\phi}
  \end{align*}
  The set $\BV{\hybridprogram}$ of bound variables of hybrid program $\hybridprogram$, i.e., those may potentially be written to, is defined inductively as:
  \begin{align*}
    \BV{x:=\theta} = \BV{x:=*}  & = \{ x \} \\
     \BV{?\phi} & = \emptyset \\
    \BV{x'=\theta \,\& \,\phi}  & =   \{ x, x^\prime \}    \\
    \BV{\alpha ; \beta} = \BV{\alpha \cup \beta} & =  \BV{\alpha} \cup \BV{\beta}                    \\
    \BV{\alpha^*} &  = \BV{\alpha} 
  \end{align*}  
\end{definitions}

\begin{definitions}[Must-bound variables] The set $\MBV{\hybridprogram}$ $\subseteq$ $\BV{\hybridprogram}$ of most bound variables of hybrid program $\hybridprogram$, i.e., all those that must be written to on all paths of $\hybridprogram$, is defined inductively as:
  \begin{align*}
    \MBV{x:=\theta} = \MBV{x:=*}  & = \{ x \} \\
     \MBV{?\phi} & = \emptyset \\
    \MBV{x'=\theta \,\& \,\phi}  & =   \{ x, x^\prime \}    \\
    \MBV{\alpha \cup \beta} & =  \MBV{\alpha} \cap \MBV{\beta}                    \\
    \MBV{\alpha ; \beta} & =  \MBV{\alpha} \cup \MBV{\beta}                    \\
    \MBV{\alpha^*} &  = \emptyset
  \end{align*}  
\end{definitions}

\begin{definitions}[Free variables]
  The set $\FV{\theta}$ of variables of term $\theta$ is defined inductively as:
   \begin{align*}
     \FV{x} & = \{x\} \\
     \FV{c} & = \emptyset \\
     \FV{\theta_1 \oplus \theta_2} & = \FV{\theta_1} \cup \FV{\theta_2} ~~~\oplus \in \{ +, \times \}
   \end{align*}
  The set $\FV{\phi}$ of free variables of \differentiallogic formula $\phi$ is defined inductively as:
   \begin{align*}
     \FV{\theta_1 \sim \theta_2} & = \FV{\theta_1} \cup \FV{\theta_2} \\
     \FV{\neg \phi} & = \FV{\phi} \\
     \FV{\phi \lor \psi}  = \FV{\phi \land \psi}    &  = \FV{\phi} \cup \FV{\psi} \\
     \FV{\phi \rightarrow \psi}   &  = \FV{\phi} \cup \FV{\psi} \\
     \FV{\forall x.~ \phi} = \FV{\exists x.~ \phi}   &  = \FV{\phi} \setminus \{ x \} \\
     \FV{\HPbox{\alpha}\phi}   &  = \FV{\alpha} \cup (\FV{\phi} \setminus \MBV{\alpha})
   \end{align*}
   The set $\FV{\hybridprogram}$ of bound variables of hybrid program $\hybridprogram$ is defined inductively as:
   \begin{align*}
     \FV{x:=\theta}  & = \FV{\theta} \\
     \FV{x:=*}     & = \emptyset   \\
     \FV{?\phi} & = \FV{\phi} \\
     \FV{x'=\theta \,\& \,\phi}  & =   \{ x \} \cup \FV{\theta} \cup \FV{\phi}   \\
     \FV{\alpha \cup \beta} & =  \FV{\alpha} \cup \FV{\beta} \\
     \FV{\alpha ; \beta} & =  \FV{\alpha} \cup (\FV{\beta} \setminus \MBV{\alpha})                    \\
     \FV{\alpha^*} &  = \FV{\alpha} 
   \end{align*}  
 \end{definitions}
 \begin{definitions}[Variable sets] The set $\variableSet{\hybridprogram}$, variables of hybrid program $\hybridprogram$ is $\BV{\hybridprogram} \cup \FV{\hybridprogram}$. The set $\variableSet{\phi}$, variables of \differentiallogic formula $\phi$ is $\BV{\phi} \cup \FV{\phi}$. 
 \end{definitions}

\section{Proofs} \label{appendix:proofs}
\begin{proof}[\textbf{Proof of Theorem~\ref{lemma:robust-safety}}]
 $\eqprogram{\programwoattack}{\attacked{\programwoattack}{S_A}}{\FV{\phi_{pre} \land \phi_{post}}}$ means for any execution $\tracesym^q$ of $\attacked{\programwoattack}{S_A}$, there exists an execution $\tracesym^p$ of $\programwoattack$ that agrees on $\FV{\phi_{pre} \land \phi_{post}}$ at the starting state and the end of every control iteration. That means if the starting state of $\tracesym^q$ satisfies $\phi_{pre}$,  the starting state  of $\tracesym^p$ satisfies $\phi_{pre}$ (Lemma 3 from \cite{platzer2017complete}). 
Meanwhile, since $\HPproperty{\phi_{pre}}{[\programwoattack]}{\phi_{post}}$, the last state of $\tracesym^p$ satisfies $\phi_{post}$. And last states of $\tracesym^q$ and $\tracesym^p$ agree on free variables used in $\phi_{post}$, so the last state of $\tracesym^q$ satisfies $\phi_{post}$ (Lemma 3 from \cite{platzer2017complete}).
$\HPproperty{\phi_{pre}}{[\attacked{\programwoattack}{S_A}]}{\phi_{post}}$ holds.
\end{proof}

\begin{proof}[\textbf{Proof of Property~\ref{theorem:eq-self} to \ref{theorem:eq-unmodified} of Theorem~\ref{theorem:loop-free-properties}}]
  By the definition of $\Hequivalence$. 
\end{proof}

\begin{lemma} \label{lemma:eq-state-basic}
  $\Hequivalence$ of states is transitive, reflective, and symmetric.
\end{lemma}
\begin{proof} By the definition of $\worldloweq{}{}{\Hsymbol}$. \end{proof}



\begin{lemma} \label{theorem:eq-state-eq}
  For program $P$, state $\traceState$, $\traceState'$, $\traceStateprime$, and set $\Hsymbol$ such that
  $\HPtransitionPair{\traceState}{\traceStateprime} \in \HPtransition{P}$, $\worldloweq{\traceState}{\traceState'}{\Hsymbol}$, and $\FV{P} \subseteq \Hsymbol$, then there exists $\traceStateprime'$ such that $\HPtransitionPair{\traceState'}{\traceStateprime'} \in \HPtransition{P}$ and $\worldloweq{\traceStateprime}{\traceStateprime'}{\Hsymbol}$. 
\end{lemma}
\begin{proof}
  By the definition of $\worldloweq{}{}{\Hsymbol}$ and lemma~4 from \cite{platzer2017complete}. \end{proof}  

\begin{proof}[\textbf{Proof of Property~\ref{theorem:eq-decomposition} of Theorem~\ref{theorem:loop-free-properties}}]
  We prove that for any execution of $A ; C$, there exists an execution of $B ; D$ that agrees with it on $\Hsymbol$. The other direction can be proven similarly.
  For any execution $\tracesym^{ac}$ of $A ; C$, let $\traceState_{ac_f}$ and $\traceState_{ac_l}$ be its first and last state respectively. Then there exists a state $\traceState_{ac_m}$ such that $\HPtransitionPair{\traceState_{ac_f}}{\traceState_{ac_m}}$ $\in$ $\HPtransition{A}$ and $\HPtransitionPair{\traceState_{ac_m}}{\traceState_{ac_l}}$ $\in$ $\HPtransition{C}$. 
  Since $\worldloweq{A}{B}{\Hsymbol}$, there exist state $\traceState_{b_f}$, $\traceState_{b_l}$ such that $\HPtransitionPair{\traceState_{b_f}}{\traceState_{b_l}}$ $\in$ $\HPtransition{B}$,  $\worldloweq{\traceState_{b_f}}{\traceState_{ac_f}}{\Hsymbol}$, and $\worldloweq{\traceState_{b_l}}{\traceState_{ac_m}}{\Hsymbol}$. Likewise, since $\worldloweq{C}{D}{\Hsymbol}$, there exists state $\traceState_{d_f}$, $\traceState_{d_l}$ such that $\HPtransitionPair{\traceState_{d_f}}{\traceState_{d_l}}$ $\in$ $\HPtransition{D}$,  $\worldloweq{\traceState_{d_f}}{\traceState_{ac_m}}{\Hsymbol}$, and $\worldloweq{\traceState_{d_l}}{\traceState_{ac_l}}{\Hsymbol}$. By transitivity (Lemma~\ref{lemma:eq-state-basic}), we get  $\worldloweq{\traceState_{b_l}}{\traceState_{d_f}}{\Hsymbol}$.
  Since $\FV{D} \subseteq \Hsymbol$, by Lemma~\ref{theorem:eq-state-eq}, there exist state $\traceState_{d'_l}$ such that $\HPtransitionPair{\traceState_{b_l}}{\traceState_{d'_l}}$ $\in$ $\HPtransition{D}$ and $\worldloweq{\traceState_{d_l}}{\traceState_{d'_l}}{\Hsymbol}$. Since $\worldloweq{\traceState_{ac_l}}{\traceState_{d'_l}}{\Hsymbol}$ and $\worldloweq{\traceState_{ac_f}}{\traceState_{b_f}}{\Hsymbol}$ (by transitivity), for the execution of $A; C$ from $\traceState_{ac_f}$ to  $\traceState_{ac_l}$, we have
  $\HPtransitionPair{\traceState_{b_f}}{\traceState_{d'_l}}$ $\in$ $\HPtransition{B; D}$, $\worldloweq{\traceState_{ac_f}}{\traceState_{b_f}}{\Hsymbol}$, and $\worldloweq{\traceState_{ac_l}}{\traceState_{d'_l}}{\Hsymbol}$. $\worldloweq{A ; C}{ B ; D}{\Hsymbol}$ holds.   
\end{proof}

\begin{proof}[\textbf{Proof of Property~\ref{theorem:eq-loop} of Theorem~\ref{theorem:loop-free-properties}}]
By induction on the number of iterations of $\alpha^*$ and $\beta^*$. Base case is trivial. For the induction case, assume $\worldloweq{\alpha^k}{\beta^k}{\Hsymbol}$ is true, we can prove $\worldloweq{\alpha^k; \alpha}{\beta^k; \beta}{\Hsymbol}$ using  Property~\ref{theorem:eq-decomposition} by letting $A$ be $\alpha^k$, $B$ be $\beta^k$, $C$ be $\alpha$, and $D$ be $\beta$. Thus, $\worldloweq{\alpha^*}{\beta^*}{\Hsymbol}$ holds.
\end{proof}

\begin{proof}[\textbf{Proof of robust safety of the ABS model}]
Let $\programwoattack$ be the hybrid program modeling ABS with duplicated sensors. Assume sensor $\omega_1$ is compromised.
Let $A$ be the \inlineCode{voting} program, $B$ be $\attacked{A}{\{\omega_1\}}$, and $C$ be program $\programwoattack$ with \inlineCode{voting} excluded (i.e., $\programwoattack=(A ; C)^*$ and $\attacked{\programwoattack}{\{\omega_1\}} = (B ; C)^*$). Here, 
$\FV{A}$ = $\FV{B}$ = $\{ \omega_p \}$, $\FV{A; C} = \FV{B; C}$, and $\FV{C}$ = $\{ \omega_s \}$ $\cup$ $\FV{A; C}$. 

By the definition of $\eqprogram{}{}{\Hsymbol}$, 
$\eqprogram{A}{B}{\{\omega_s, \omega_p \}}$ holds, which means
\[
  \eqprogram{A}{B}{\FV{C}} \tag*{(Property~\ref{theorem:eq-unmodified})}
\]
With $\eqprogram{C}{C}{\FV{C}}$ (Property~\ref{theorem:eq-self}), we get
\[
  \eqprogram{(A; C)}{(B; C)}{\FV{C}} \tag*{(Property~\ref{theorem:eq-decomposition})}
\]
which leads to
\[
  \eqprogram{(A; C)^*}{(B; C)^*}{\FV{C}} \tag*{(Property~\ref{theorem:eq-loop})}
\]
Property~\ref{theorem:eq-subset} also applies to programs with loop, and $\{\omega_p, v_p\}$ $\subseteq$ $\FV{C}$, thus
\[
  \eqprogram{(A; C)^*}{(B; C)^*}{\{\omega_p, v_p\}} \tag*{(Property~\ref{theorem:eq-subset})}
\]
Since $\FV{\phi_{pre} \land \phi_{post}}$ = $\{\omega_p, v_p\}$, we have $\robustsafety{\programwoattack}{\phi_{pre}}{\phi_{post}}{\{\omega_{s_1}\}}$ (Theorem~\ref{lemma:robust-safety}).

Similarly, we can prove $\robustsafety{\programwoattack}{\phi_{pre}}{\phi_{post}}{\{\omega_{s_2}\}}$ and $\robustsafety{\programwoattack}{\phi_{pre}}{\phi_{post}}{\{\omega_{s_3}\}}$.
\end{proof}

\medskip

\newcommand{\fvp}{\texttt{fv}}

\begin{proof}[\textbf{Proof of robust safety of Boeing 737-MAX model}]
 Let $A$ be program \inlineCode{ctrl$_{aoa}^\prime$},  $B$ be program $\attacked{A}{\{ s_{L} \}}$, $C$ be program \inlineCode{MCAS($aoa$); plant} in Figure~\ref{fig:hp-boeing737max}.
Here, $\FV{A}$ = $\FV{B}$ = $\{ aoa_p \}$, let $\fvp$ be the set of free variables of program $A; C$, then $\FV{B; C}$=$\fvp$, and $\FV{C}$ would be $\{ aoa_s \} \cup \fvp$. 
We can prove $\robustsafety{A; C}{\phi_{pre}}{\phi_{post}}{\{ s_L\}}$ with the following steps: 

By definition of $\eqprogram{}{}{\Hsymbol}$, we prove
$\eqprogram{A}{B}{\{aoa_s, aoa_p \}}$, which means
\[
  \eqprogram{A}{B}{\{ aoa_s, aoa_p \} \cup \fvp} \tag*{(Property~\ref{theorem:eq-unmodified})}
\]
With $\eqprogram{C}{C}{\{ aoa_s, aoa_p \} \cup \fvp}$ (Property~\ref{theorem:eq-self}), we know
\[
  \eqprogram{A; C}{B; C}{\{ aoa_s, aoa_p \} \cup \fvp} \tag*{(Property~\ref{theorem:eq-decomposition})}
\]
 Since $\FV{A; C} \cup \FV{B; C}$ $\subseteq$ ${\{ aoa_s, aoa_p \} \cup \fvp}$, we know
 \[
   \eqprogram{(A; C)^*}{(B; C)^*}{\{ aoa_s, aoa_p \} \cup \fvp} \tag*{(Property~\ref{theorem:eq-loop})}
 \]
 Property~\ref{theorem:eq-subset} applies to programs with loop as well, so 
  \[
   \eqprogram{(A; C)^*}{(B; C)^*}{\fvp} \tag*{(Property~\ref{theorem:eq-subset})}
 \]
 Since formula $\phi_{pre}$ and $\phi_{post}$ typically refer to free variables in $\fvp$, we get $\robustsafety{\programwoattack}{\phi_{pre}}{\phi_{post}}{\{ s_{L}\}}$ holds.  (Theorem~\ref{lemma:robust-safety}). Similarly, we can prove $\robustsafety{\programwoattack}{\phi_{pre}}{\phi_{post}}{\{ s_{R}\}}$.
 \end{proof}



\section{Limitations of the Self-Composition Approach} \label{appendix:limitation}

One limitation of our self-composition approach is that it applies only for hybrid programs that have \emph{total} semantics for all low-integrity inputs. 
It means if a program has a \emph{valid} execution on an input state $\traceState$ (i.e., exist state $\traceStateprime$ such that $\HPtransitionPair{\traceState}{\traceStateprime} \in \HPtransition{\hybridprogram}$), then for any state $\traceState^\prime$ that differs with $\traceState$ only in low-integrity inputs, there exists  $\traceStateprime^\prime$ that $\HPtransitionPair{\traceState}{\traceStateprime} \in \HPtransition{\attacked{\hybridprogram}{S_A}}$.

A program may have \emph{partial} (not total) semantics on low-integrity inputs for two reasons: (1) some low-integrity inputs fail test conditions in all execution paths, for example, if $a$ is a low-integrity variable, then $?\mathit{a > 0}$ is a program whose semantics are partial on low-integrity inputs; (2) the program's evolution constraint depends on low-integrity inputs. For example, if $a$ is a low-integrity variable, $(x'=\theta\& \mathit{a > 0})$ is a program whose semantics are partial on low-integrity inputs. 


Fortunately, there is a relatively simple way to check that hybrid programs meet this requirement. First, given a set of low-integrity sensor variables, a straightforward program analysis can identify all variables that might depend on a low-integrity sensor variables; call these the low-integrity variables. Second, check that all evolution constraints do not include any low-integrity variables. Third, check that any test $?\phi_i$ that includes a low-integrity variable occurs as part of a construct $?\phi_1;\alpha_1 \cup \dots \cup ?\phi_n;\alpha_n$ such that $\phi_1 \vee \dots \vee \phi_n$ is valid (i.e., the tests are exhaustive and so at least one of the branches of the nondeterministic choice will be true).

Well-designed hybrid program models should have total semantics on low-integrity inputs, except in specific situations that rarely depend on low-integrity sensor variables. 
Models that do not have total semantics on low-integrity inputs typically do not correspond to actually implementable control strategies, and are therefore only vacuously safe.

\section{Soundness Proof of the Self-Composition Approach} \label{appendix:sound-proof}

We use trace semantics of hybrid programs~\cite{platzer2007temporal, jeannin2014dtl} to prove Theorem~\ref{theorem:soundness-composition}. The \emph{trace semantics} of hybrid programs assigns to each program $\alpha$ a set of traces $\tau(\alpha)$. A state is a map from the set of variables to real numbers. The set of all variables is denoted $\allvariableSet$. The set of all states is denoted $\stateSet$. A separate state $\Lambda$ (not in $\stateSet$) denoting a failure of the system. 

A \emph{trace} is a (non-empty) finite or infinite sequence $\tracesym = (\tracefun_0, \tracefun_1, ... )$ of \emph{trace functions} $\tracefun_i$ : $[0, r_i] \rightarrow \stateSet$ with duration $r_i \in \mathbb{R}$. 
A position of $\tracesym$ is a pair $(i, \iota)$ with $i \in \mathbb{N}$ and $\iota$ in the interval $[0, r_i]$; the state of $\tracesym$ at $(i, \iota)$ is $\tracefun_i^{\iota}$.
For a state $\traceState \in \stateSet$, $\hat{\traceState}$: $0 \mapsto \traceState$ is a point flow at $\traceState$ with duration 0.
A trace terminates if it is a finite sequence $\tracesym = (\tracefun_0, \tracefun_1, ... \tracefun_n)$ and $\tracefun_n \not= \Lambda$. In that case, the last state is denoted as $\sigma_n(r_n)$. The first state of $\tracesym$, denoted $\traceTop{\tracesym}$, is  $\tracesym_0(0)$. The set of all traces is $\traceSet$.

We denote by $\traceState[x \mapsto r]$ the valuation assigning variable $x$ to $d \in \mathbb{R}$ and matching with $\traceState$ on all other variables. 

The trace semantics $\tau(\alpha)$ of a hybrid program $\alpha$ is defined inductively\cite{jeannin2014dtl}:
\begin{itemize}[noitemsep]
\item $\tau(x :=\theta) = \{ (\hat{\traceState}, \hat{\traceStateprime}) ~|~ \traceStateprime = \traceState[x \mapsto \HPtermsem{\traceState}{\theta}] \}$;
\item $\tau(x^{\prime} =\theta \& \evolconstraint) = \{ (\tracefun) : \tracefun$ ~is a state flow of order 1 defined on [0, r] or [0, +$\infty$] solution of $x^{\prime}=\theta$, and for all $t$ in its domain, $\tracefun(t) \models \evolconstraint\}$ $\cup$ $\{(\hat{\traceState}, \hat{\Lambda}) : \traceState \not\models \evolconstraint\}$;
\item $\tau(?\phi) = \{ (\hat{\traceState}) ~|~ \traceState \models \phi \}$ $\cup$ $\{(\hat{\traceState}, \hat{\Lambda}) : \traceState \not\models \phi\}$;
\item $\tau(\alpha \cup \beta) = \tau(\alpha) \cup \tau(\beta)$;
\item $\tau(\alpha; \beta) = \{ \tracesym \circ \rho\ : \tracesym \in \tau(\alpha), \rho \in \tau(\beta)$ when $\tracesym \circ \rho \text{~is defined} \}$;
  where the composition $\tracesym \circ \rho$ of $\tracesym = (\tracefun_0, ..., \tracefun_n)$ and $\rho = (\rho_0, ..., \rho_m)$ is
  \begin{itemize}
  \item $\tracesym \circ \rho = (\tracefun_0, ..., \tracefun_n, \rho_0, ..., \rho_m)$ if $\tracesym$ terminates and $\traceEnd{\tracesym} = \traceTop{\rho}$;
  \item $\tracesym$ if $\tracesym$ does not terminate;
  \item undefined otherwise;
    \end{itemize}
\item $\tau(\alpha^*) = \cup_{n\in \mathbb{N}} \tau(\alpha^n)$, where $\alpha^0$ is defined as $?true$, $\alpha^1$ is defined as $\alpha$ and $\alpha^{n+1}$ is defined as $\alpha^n; \alpha$ for $n \geq 1$; 
\item $\tau(x :=*) = \{ (\hat{\traceState}, \hat{\traceStateprime}) ~|~ \traceStateprime = \traceState[x \mapsto d] \}$ where $d$ is some real value.
\end{itemize}

Notice that the trace semantic for $\tau(x :=*)$ is not defined in  \cite{platzer2007temporal, jeannin2014dtl}. We add it to complete the definition of trace semantic needed in this work. 

We refer to finite traces that end with failure state $\Lambda$ as \emph{failure traces}, and other traces as \emph{normal traces}. We denote $\normaltraces{\program}$ the set of normal traces of a program $\program$:
\[
  \normaltraces{\program} = \{ \tracesym \in \tau(\program) ~|~ \traceEnd{\tracesym} \not= \Lambda \lor \tracesym \text{~does not terminate}\}
\]


Now, we formalize the $\Hequivalence$ of states, trace functions, traces, and programs.
Compare with Definition~\ref{def:equivlance}, these formal definitions are more general (can be applied on programs with different variable sets) and uses a mapping function between variables in two states (instead of using just a set). 

\begin{definitions}[$\Hequivalence$ of states] \label{def:low-eq-state} We define $\worldloweqfun{\traceState_i}{\traceState_j}{\varfun}$, for states $\traceState_i$ and $\traceState_j$ that agree on corresponding variables that are related by function $\varfun$, i.e., 
   \begin{align*}
    & \forall x \in dom(\varfun), \traceState_i(x) = \traceState_j(\varfun(x)) 
  \end{align*}
\end{definitions}
Here the domain of $\varfun$ corresponds to the $\Hsymbol$ in Definition~\ref{def:equivlance}. And $\varfun$ is often a subset of the renaming function of the program of concern.

\begin{definitions}[$\Hequivalence$ of trace functions] \label{def:low-eq-trace-fun} 
  We define $\worldloweqfun{\tracefun_i}{\tracefun_j}{\varfun}$, for trace functions $\tracefun_i$ and $\tracefun_j$ that have the same domain and $\Hequivalent$ states at all domain values:
  \[ dom(\tracefun_i) = dom (\tracefun_j)
    \text{~and~}
    \forall p \in dom(\tracefun_i), \worldloweqfun{\tracefun_i(p)}{\tracefun_j(p)}{\varfun} \]
\end{definitions}

\begin{definitions}[$\Hequivalence$ of traces] \label{def:low-eq-trace} We define $\worldloweqfun{\tracesym^a}{\tracesym^b}{\varfun}$, for trace $\tracesym^a$ and $\tracesym^b$ that agree on (1) the \emph{first} state (2) trace functions whose domains are not [0, 0], and (3) the \emph{last} state if both are finite traces. Figure~\ref{fig:definition-Heq-trace} shows the formal definition.
\end{definitions}
\begin{figure}
\begin{mathpar}
  \inferrule[single.function]
  {
    \worldloweqfun{\tracesym^a_0}{\tracesym^b_0}{\varfun}
  }
  {
    \worldloweqfun{(\tracesym^a_0)}{(\tracesym^b_0)}{\varfun}
  }
  
  \inferrule[trace-plant]
  { m \geq 1
    \\
    n \geq 1
    \\\\
    \worldloweqfun{\tracesym^a_0}{\tracesym^b_0}{\varfun}
    \\
    dom(\tracesym^a_0) \not= [0,0]
    \\
    dom(\tracesym^b_0) \not= [0,0]
    \\\\
    \worldloweqfun{(\tracesym^a_1 \dots \tracesym^a_m)}{(\tracesym^b_1 \dots \tracesym^b_n)}{\varfun}
  }
  {
    \worldloweqfun{(\tracesym^a_0 \dots \tracesym^a_m}{(\tracesym^b_0 \dots \tracesym^b_n)}{\varfun} 
  }
  
  \inferrule[trace-discrete]
  { (m \geq 0 \land n \geq 1 ) \lor  (m \geq 1 \land n \geq 0 )
    \\  \worldloweqfun{\traceTop{\tracesym^a}}{\traceTop{\tracesym^b}}{\varfun}
    \\\\
    dom(\tracesym^a_p) \not= [0,0] \lor p = m
    \\\\
    dom(\tracesym^b_q) \not= [0,0] \lor q = n
    \\\\
    \forall i (0 \leq i < p \land p < m), dom(\tracesym^a_i) = [0,0]
    \\
    \forall j (0 \leq j < q \land q < n), dom(\tracesym^b_j) = [0,0]
    \\\\
    \worldloweqfun{(\tracesym^a_p \dots \tracesym^a_m)}{(\tracesym^b_q \dots \tracesym^b_n)}{\varfun}
  }
  {
    \worldloweqfun{(\tracesym^a_0 \dots \tracesym^a_{p-1}, \tracesym^a_p \dots \tracesym^a_m)}{(\tracesym^b_0 \dots \tracesym^b_{q-1}, \tracesym^b_q \dots \tracesym^b_n)}{\varfun} 
  }
\end{mathpar}
\caption{Definition of $\Hequivalence$ of traces}
\label{fig:definition-Heq-trace}
\end{figure}

We write $\worldloweqfun{\traceState_1}{\traceState_2}{id_{\Hsymbol}}$ to mean that $\traceState_1$ and $\traceState_2$ are $\Hequivalence$ with respect to an identity function defined on set $\Hsymbol$ and undefined otherwise (i.e., $id_{\Hsymbol}$).
We write  $\worldloweqfun{\tracesym^a}{\tracesym^b}{id_{\Hsymbol}}$ to indicate traces $\tracesym^a$ and $\tracesym^b$ are equivalent on $id_{\Hsymbol}$. We write $\worldloweqfun{\tracesym^a}{\tracesym^b}{id}$ to mean that $\tracesym^a$ and $\tracesym^b$ are equivalent with \emph{all} variables, i.e., the two traces use the same set of variables and they are $\Hequivalent$.

\begin{definitions}[$\Hequivalence$ of two programs by traces] \label{def:noninterference-two}
For two hybrid programs $\program_1$ and $\program_2$ of the canonical form, a function $\varfun$ maps variables in $\program_1$ to variables in $\program_2$, $\program_1$ $\NI$ $\program_2$ is defined as follows:
\begin{align*}
   \forall \tracesym^{a} \in \traceSem{\program_1}, 
            ~\exists \tracesym^{b} \in \traceSem{\program_2}  \suchthat  \worldloweqfun{\tracesym^a}{\tracesym^b}{\varfun}                                    
\end{align*}
\end{definitions}

To help express the agreement between executions composed in a self-composition, we introduce the notion of \emph{projections} on states, trace functions, and traces. 

\begin{definitions}[Projection] \label{def:trace-state-projection}
  For state $\traceState$ and a set $V$ of variables such that $V \subseteq \variableSet{\traceState}$, the $V$ projection of state $\traceState$, denoted $\tproj{\traceState}{V}$, is a map $\{ x \mapsto \traceState(x) \}$ for all $x \in V$.

  For a trace function $\tracefun_i$: [$0, r_i$] $\rightarrow$ $\stateSet$ and a set $V$ of variables such that $V \subseteq \variableSet{\tracefun_i}$, the $V$ projection of $\tracefun_i$, denoted $\tproj{\tracefun_i}{V}$, is $\{ \iota \mapsto (\tproj{\tracefun_i(x)}{V}) \}$ for all $\iota \in dom(\tracefun_i)$.

  For a trace $\tracesym = (\tracefun_0, \dots ,\tracefun_n)$ and a set $V$ of variables such that $V \subseteq \variableSet{\tracefun}$, the $V$ projection of $\tracesym$, denoted $\tproj{\tracesym}{V}$, is computed by pointwise projecting every trace function of $\tracesym$: 
\begin{align*}
  \tproj{\tracesym}{V} = (\tproj{\tracefun_0}{V}, \dots, \tproj{\tracefun_n}{V})
\end{align*}

For a program $\program$, we write $\tproj{\tracesym}{\program}$, to mean $\tproj{\tracesym}{\variableSet{\program}}$. 
Notation $\tproj{}{\program}$ also applies to states and trace functions. 
\end{definitions}

The soundness theorem (Theorem~\ref{theorem:soundness-composition}) has a list of promises: a program $\program$ and $\program_c$ ($\program$ in canonical form), a set $S_A$ of variables, a set $\eqSet$ of variables such that $S_A \subseteq \BV{\program}$, $\eqSet \subseteq \BV{\program}$, and $S_A \cap \eqSet = \emptyset$. We assume but elide these promises in the following definitions and lemmas.

\begin{definitions}[Self-composition preserves equivalence formula]
  \label{def:preservation-composition-formalized}
  The desired property of a self-composition:
\[
  eq_{\eqSet}^{\xi} \rightarrow [\compsym{\programwoattack}{S_A}{\xi}]eq_{\eqSet}^{\xi}
\] is formalized as follows:
\begin{align*}
  \forall
  & \tracesym \in \traceSem{\compsym{\programwoattack}{S_A}{\xi}}  \suchthat \\
  &                                                           \worldloweqfun{\traceTop{(\tproj{\tracesym}{\programwoattack})}}{\traceTop{(\tproj{\tracesym}{\renaming{\programwoattack}{\xi}})}}{\varfun},  \\ 
  & ~\worldloweqfun{\tproj{\tracesym}{\programwoattack}}{\tproj{\tracesym}{\renaming{\programwoattack}{\xi}}}{\varfun} 
  \end{align*}
Where $\varfun$ is $ \{ (x, \renaming{x}{\xi}) ~|~ x \in \eqSet \}$.
\end{definitions}

\begin{assumption}[A program has total semantics on low-integrity inputs (formalized)] \label{well-formed-inputs}
\begin{align*} 
  \forall & \traceState_{1}, \traceState_{2}: \stateSet \suchthat \worldloweq{\traceState_{1}}{\traceState_{2}}{\eqSet}, \\
          & ~\exists \tracesym^a \in \normaltraces{\programwoattack} \suchthat \traceTop{\tracesym^a} = \traceState_{1} \\
          & \leftrightarrow ~\exists \tracesym^b \in \normaltraces{\attacked{\programwoattack}{S_A}} \suchthat \traceTop{\tracesym^b} = \traceState_{2} 
\end{align*}
\end{assumption}

\begin{lemma}[Renaming preserves trace] \label{lemma:rename.exists.trace}
  For a hybrid program $\programwoattack$ and a renaming function $\xi$ on $\programwoattack$:
  \[
    \forall \tracesym \in \traceSem{\programwoattack}, \renaming{\tracesym}{\xi} \in \traceSem{\renaming{\programwoattack}{\xi}}
  \]
  Where $\renaming{\tracesym}{\xi}$ is $\tracesym$ with variables renamed according to $\xi$.
\end{lemma}
\begin{proof} By induction on $\programwoattack$.  \end{proof}



\begin{lemma}[Renaming preserves trace existence] \label{lemma:renamed.trace.existence}
  \begin{align*}
    \forall & \traceState_{1}, \traceState_{2}: \stateSet \suchthat \worldloweq{\traceState_{1}}{\traceState_{2}}{\eqSet}, \\
     & \forall \tracesym \in \traceSem{\programwoattack} \suchthat \traceTop{\tracesym} = \traceState_{1},\\
    & ~\exists \tracesym' \in \traceSem{\renaming{\attacked{\programwoattack}{S_A}}{\xi}} \suchthat    \traceTop{\tracesym'} = \renaming{\traceState_{2}}{\xi}
  \end{align*}
\end{lemma}
\begin{proof} By assumption~\ref{well-formed-inputs}, a trace of $\attacked{\programwoattack}{S_A}$ exists with starting state $\traceState_{2}$. By lemma~\ref{lemma:rename.exists.trace}, we know $\renaming{\tracesym^\prime}{\xi}$ is a normal trace of $\renaming{\attacked{\programwoattack}{S_A}}{\xi}$.
\end{proof}



\begin{lemma} \label{lemma:composition-disjoint} Trace preserves after adding disjoint variable sets. 
\begin{align*}
  \forall & \traceState_{1}~ \traceState_{2} : \stateSet \suchthat \\
          & \variableSet{\traceState_{1}} = \variableSet{\program} \text{~and~} \variableSet{\traceState_{1}} \cap \variableSet{\traceState_{2}} = \emptyset \\
          & ~\forall \tracesym \in \traceSem{\program} \suchthat \traceTop{\tracesym} = \traceState_{1}, \\ 
          & ~~\exists \tracesym' \in \traceSem{\program} \suchthat \\
          & ~~~\traceTop{\tracesym'} = \joinworld{\traceState_{1}}{\traceState_{2}} \text{~and~} \tproj{\tracesym'}{\program} = \tracesym 
\end{align*}
Where $\oplus$ means the join of two non-overlapping states.
\end{lemma}
\begin{proof} 
By induction on $\program$. \end{proof}




\begin{lemma}[Projection preserves trace]
  \label{lemma:projection-trace-existance} For program  $\program$, 
\begin{align*}
  \forall \tracesym \in \traceSem{\program},  \tproj{\tracesym}{\program}  \in \traceSem{\program} \text{~and~} \worldloweqfun{\tproj{\tracesym}{\program}}{\tracesym}{id_{\variableSet{\program}}}
\end{align*}
\end{lemma}
\begin{proof} By the definition of trace semantics and projection. \end{proof}


\begin{lemma}[Projection not affected by programs with disjoint variables]
  \label{lemma:project-not-affected-disjoint} For program  $\program_1$ and $\program_2$ such that  $\BV{\program_1}$ $\cap$  $\BV{\program_2}$ = $\emptyset$, 
\begin{align*}
  \forall \tracesym \in \traceSem{\program_1},  \tproj{\traceTop{\tracesym}}{\program_2} = \tproj{\traceEnd{\tracesym}}{\program_2}
\end{align*}
\end{lemma}
\begin{proof} By induction on $\program_1$ and definition of $\tproj{}{}$. \end{proof}


\begin{lemma}[Composition preserves trace existence] \label{lemma:composition-plant} For program $\alpha = (ctrl; x'=\theta\&\evolconstraint)$, 
  \begin{align*}
\forall & \traceState_{1}, \traceState_{2} : \stateSet \suchthat \worldloweqfun{\traceState_{1}}{\traceState_{2}}{\varfun}, \\
          & \variableSet{\traceState_{1}} = \variableSet{\alpha}, \text{~and~} \variableSet{\traceState_{2}} = \variableSet{\renaming{\alpha}{\xi}}, \\
          & ~\forall \tracesym \in \traceSem{\alpha} \suchthat \traceTop{\tracesym} = \traceState_{1}, \\ 
    & ~\exists \tracesym' \in \traceSem{ctrl;\renaming{\attacked{ctrl}{S_A}}{\xi}; (x'=\theta, \renaming{x'=\theta}{\xi}) \\
          & ~~~~~~~~~~~~~~ \&(\evolconstraint \land \renaming{\evolconstraint}{\xi})} \suchthat \\
          & ~~~ \worldloweqfun{\tproj{\tracesym'}{\alpha}}{\tracesym^{a}}{id} \text{~and~} \tproj{(\traceTop{\tracesym'})}{\renaming{\alpha}{\xi}} = \traceState_{2} 
\end{align*}
Where $\varfun$ is $ \{ (x, \renaming{x}{\xi}) ~|~ x \in \eqSet \}$.
\end{lemma}
\begin{proof} Let $\tracesym^a = (\tracefun^a_0...\tracefun^a_m)$, then $(\tracefun^a_0...\tracefun^a_{m-1})$ is a trace of $ctrl$, and $\tracefun_m$ : $[0, r_1] \mapsto \stateSet$ is a trace function for $x'=\theta\&\evolconstraint$.
According to Assumption~\ref{well-formed-inputs}, there exists $\tracesym^b \in \traceSem{\renaming{\attacked{\alpha}{S_A}}{\xi}}$. 
We can then prove the part of $ctrl; \renaming{\attacked{ctrl}{S_A}}{\xi}$ by lemma~\ref{lemma:composition-disjoint}, \ref{lemma:project-not-affected-disjoint} and the definition of $\worldloweqfun{}{}{}$. For the plant part, we know (by Assumption~\ref{well-formed-inputs}) low-integrity values cannot affect evolution constraints, meaning input states $\traceState_{1}$ and $\traceState_{2}$ should be able to last the same duration of evolution. 
Thus, for any duration $r_1$ that trace $\tracesym^a$ has, the duration of the other trace $\tracefun^b$ can match it, i.e., $r_1$ = $r_2$. 
Thus there exist a trace function $[0, r_1] \mapsto \stateSet$: $x \mapsto \joinworld{\tracesym_{1}(x)}{(\tproj{\tracesym_{2}(x)}{\BV{\renaming{\alpha}{\xi}}} )}$ for the composed dynamic $(x'=\theta, \renaming{x'=\theta}{\xi})\&(\evolconstraint \land \renaming{\evolconstraint}{\xi})$, whose $\alpha$ projection is indistinguishable from $\tracesym^a$.
Combined with the result for $ctrl; \renaming{\attacked{ctrl}{S_A}}{\xi}$, this lemma is proven. 
\end{proof}

\begin{lemma}[Assigning the same value to connected variables preserves equivalence] \label{lemma:assignH.preserve.eq.state} 
\begin{align*}
  \forall & \traceState_{1}, \traceState_{2} : \stateSet \suchthat
          \worldloweqfun{\traceState_{1}}{\traceState_{2}}{\varfun}, \\
          & \forall x : \allvariableSet, ~d : \mathbb{R} \suchthat x \in dom(\varfun), \\
          & ~\worldloweqfun{\traceState_{1}[x \mapsto d]}{\traceState_{2}[\xi(x) \mapsto d]}{\varfun}
\end{align*}
\end{lemma}


\begin{lemma}[Assigning arbitrary values to non-connected variables preserves equivalence] \label{lemma:assignL.preserve.eq.state}
\begin{align*}
  \forall & \traceState_{1}, \traceState_{2} : \stateSet \suchthat 
          \worldloweqfun{\traceState_{1}}{\traceState_{2}}{\varfun}, \\
          & \forall x : \allvariableSet,~ d_1, d_2 : \mathbb{R} \suchthat x \not\in dom(\varfun)  \\
          &                                                        ~~\worldloweqfun{\traceState_{1}[x \mapsto d_1]}{\traceState_{2}[\xi(x) \mapsto d_2]}{\varfun}
\end{align*}
\end{lemma}
Lemma~\ref{lemma:assignH.preserve.eq.state} and \ref{lemma:assignL.preserve.eq.state} can be proven by the definition of $\worldloweqfun{}{}{\varfun}$.


\begin{lemma}[Choice part of composition preserve equivalence] For program $choices$ that consists of non-deterministic assignments of choice variables, \label{lemma:composition-choice}
  \begin{align*}
  & \forall \traceState_{1}, \traceState_{2} : \stateSet \suchthat \worldloweqfun{\traceState_{1}}{\traceState_{2}}{\varfun}, \\
  &        ~~\forall \tracesym \in \traceSem{choices} \suchthat \traceTop{\tracesym} = \traceState_{1}, \\ 
  &        ~~~\exists \tracesym' \in \traceSem{choices; \subs{choices}{\xi}} \suchthat \\
  &        ~~~~ \worldloweqfun{\tproj{\tracesym'}{choices}}{\tracesym}{id} \text{~and~} \\
  &        ~~~~~ \worldloweqfun{\tproj{\traceEnd{\tracesym'}}{choices}}{\tproj{\traceEnd{\tracesym'}}{\renaming{choices}{\xi}}}{\varfun}
\end{align*}
\end{lemma}
\begin{proof} 
Let $\tracesym$ = $(\tracefun_0 \dots \tracefun_p)$ be the trace of program $choices$, then there exists a trace $\tracesym^b$ for $\renaming{choices}{\xi}$ with the same length as $\tracesym$, i.e., $\tracesym^b = (\tracefun^b_0 \dots \tracefun^b_p)$. We can then get a trace for program $\subs{choices}{\xi}$ by altering corresponding variables in the state. Then by lemma~\ref{lemma:assignH.preserve.eq.state} and \ref{lemma:assignL.preserve.eq.state} and induction on the number of assignments in $choices$. \end{proof}


\begin{lemma}[Completeness of a single iteration] \label{lemma:composition-single}
 Let program $\programwoattack$ = $\alpha_p^*$ and $\compsym{P}{S_A}{\xi}$ = $\alpha_c^*$, 
\begin{align*}
  \forall & \traceState_{1}, \traceState_{2} : \stateSet \suchthat \worldloweqfun{\traceState_{1}}{\traceState_{2}}{\varfun}, \\
          & \variableSet{\traceState_{1}} = \variableSet{\programwoattack}, \text{~and~} \variableSet{\traceState_{2}} = \variableSet{\renaming{\programwoattack}{\xi}}, \\
          & ~\forall \tracesym \in \traceSem{\alpha_p} \suchthat \traceTop{\tracesym} = \traceState_{1}, \\ 
          & ~~\exists \tracesym' \in \traceSem{\alpha_c} \suchthat            \worldloweqfun{\tproj{\tracesym'}{\programwoattack}}{\tracesym}{id} \text{~and~} \\
  &~~~\tproj{\traceTop{\tracesym'}}{\renaming{\programwoattack}{\xi}} = \traceState_{2}
\end{align*}
\end{lemma}
\begin{proof} By Lemma~\ref{lemma:composition-plant} and \ref{lemma:composition-choice}. \end{proof}


\begin{lemma}[Projections of a sequence] \label{lemma:proj-sequence-loweq} 
\begin{align*}
  \forall & 
          \tracesym \in \traceSem{\alpha; \beta} \suchthat \BV{\alpha} \cap \BV{\beta} = \emptyset, \\
          & \exists \tracesym^a \in \traceSem{\alpha}, \tracesym^b \in \traceSem{\beta} \suchthat \\
          & ~ \worldloweqfun{\tproj{\tracesym}{\alpha}}{\tracesym^{a}}{id} \text{~and~}
            \worldloweqfun{\tproj{\tracesym}{\beta}}{\tracesym^{b}}{id}
\end{align*}
\end{lemma}
\begin{proof} By induction on $\alpha$, $\beta$, and definition of projection. \end{proof}


\begin{lemma}[Soundness of the trace for composed plant] \label{lemma:composition-plant-split} For program $\alpha$ = $(ctrl; x'=\theta \& \evolconstraint)$ and $\beta$ = $(\renaming{\attacked{ctrl}{S_A}}{\xi}; \renaming{x'=\theta}{\xi} \& \renaming{\evolconstraint}{\xi})$,
\begin{align*}         
  \forall &\tracesym \in \traceSem{ctrl; \renaming{\attacked{ctrl}{S_A}}{\xi}; \\
          &  ~~~~~~(x'=\theta, \renaming{x'=\theta}{\xi}\&(\evolconstraint \land \renaming{\evolconstraint}{\xi})} \\
          & \exists \tracesym^{a} \in \traceSem{\alpha}, \tracesym^{b} \in \traceSem{\beta} \suchthat  \\
          & ~ \worldloweqfun{\tproj{\tracesym}{\alpha}}{\tracesym^{a}}{id} \text{~and~}
            \worldloweqfun{\tproj{\tracesym}{\beta}}{\tracesym^{b}}{id}
\end{align*}
\end{lemma}
\begin{proof} By definition of trace semantics and Lemma~\ref{lemma:proj-sequence-loweq}. 
\end{proof}


\begin{lemma}[Soundness of single iteration] \label{lemma:composition-both-valid-single} Let program $\programwoattack$ = $\alpha_p^*$, $\renaming{\attacked{\programwoattack}{S_A}}{\xi}$ = $\alpha_q^*$, and $\compsym{P}{S_A}{\xi}$ = $\alpha_c^*$, 
\begin{align*}
  \forall & \tracesym \in \traceSem{\alpha_c}, \\
          &  \exists \tracesym^a \in \traceSem{\alpha_p}, \tracesym^b \in \traceSem{\alpha_q} \suchthat \\
          & ~~\worldloweqfun{(\tproj{\tracesym}{\alpha_p})}{\tracesym^a}{id} \text{~and~}
\worldloweqfun{(\tproj{\tracesym}{\alpha_q})}{\tracesym^b}{id}
\end{align*}
\end{lemma}
\begin{proof} By Lemma~\ref{lemma:proj-sequence-loweq} and \ref{lemma:composition-plant-split}. 
\end{proof}


\begin{lemma}[Renaming preserve equivalence] \label{lemma:renaming}
\[
  \programwoattack \NI \renaming{\attacked{\programwoattack}{S_A}}{\xi} \leftrightarrow
  \worldloweq{\program}{\attacked{\programwoattack}{S_A}}{dom(\varfun)}
\]
\end{lemma}
\begin{proof} By induction on the variables in program $\program$\cite{barthe2004secure}.
\end{proof}



\begin{proof}[\textbf{Proof of Theorem~\ref{theorem:soundness-composition}}] By Definition~\ref{def:preservation-composition-formalized}, Lemma~\ref{lemma:projection-trace-existance}, \ref{lemma:composition-single}, \ref{lemma:composition-both-valid-single}, \ref{lemma:renaming}, and induction on the number of iterations, we get 
$\worldloweq{\program}{\attacked{\programwoattack}{S_A}}{\eqSet}$. Since $\Hsymbol \subseteq \eqSet$, $\worldloweq{\program}{\attacked{\programwoattack}{S_A}}{\Hsymbol}$ (Property~\ref{theorem:eq-subset}).
\end{proof}


\end{document}